\RequirePackage[british]{babel}
\documentclass[reqno,a4paper,12pt]{amsart}
\usepackage{a4wide}
\usepackage[utf8x]{inputenc}
\usepackage{amsmath,amssymb,amstext,amsthm,amscd,mathrsfs,eucal}
\usepackage{graphicx,color}
\usepackage{array}
\usepackage[noadjust]{cite}
\usepackage{hyperref}
\hypersetup{%
  pdftitle   = {Lorentzian Lie 3-algebras and their Bagger--Lambert moduli space},
  pdfkeywords = {lorentzian, M2, M Theory, Bagger-Lambert, Lie},
  pdfauthor  = {Paul de Medeiros, José Figueroa-O'Farrill, Elena Méndez-Escobar},
  pdfcreator = {\LaTeX\ with package \flqq hyperref\frqq}
}
\PrerenderUnicode{éÉ}
\let\into\hookrightarrow
\newcommand{\half}{\tfrac12}
\newcommand{\fg}{\mathfrak{g}}
\newcommand{\fa}{\mathfrak{a}}

\newcommand{\fe}{\mathfrak{e}}
\newcommand{\ff}{\mathfrak{f}}
\newcommand{\fW}{\mathfrak{W}}
\newcommand{\fgl}{\mathfrak{gl}}
\newcommand{\fh}{\mathfrak{h}}
\newcommand{\fk}{\mathfrak{k}}
\newcommand{\fp}{\mathfrak{p}}
\newcommand{\fs}{\mathfrak{s}}
\newcommand{\fso}{\mathfrak{so}}
\newcommand{\fosp}{\mathfrak{osp}}
\newcommand{\fsp}{\mathfrak{sp}}
\newcommand{\fsu}{\mathfrak{su}}
\newcommand{\fu}{\mathfrak{u}}

\newcommand{\GL}{\mathrm{GL}}
\newcommand{\RR}{\mathbb{R}}

\newcommand{\ZZ}{\mathbb{Z}}
\newcommand{\eL}{\mathscr{L}}
\newcommand{\eM}{\mathscr{M}}
\newcommand{\eV}{\mathscr{V}}
\newcommand{\be}{\boldsymbol{e}}
\newcommand{\bx}{\boldsymbol{x}}
\newcommand{\by}{\boldsymbol{y}}
\newcommand{\bz}{\boldsymbol{z}}
\DeclareMathOperator{\Aut}{Aut}
\DeclareMathOperator{\Ca}{Cartan}
\DeclareMathOperator{\Iso}{Iso}
\DeclareMathOperator{\Hom}{Hom}
\DeclareMathOperator{\Der}{Der}
\DeclareMathOperator{\Ad}{Ad}
\DeclareMathOperator{\ad}{ad}
\DeclareMathOperator{\rank}{rank}
\theoremstyle{plain}
\newtheorem{lemma}{Lemma}
\newtheorem{proposition}[lemma]{Proposition}
\newtheorem{theorem}[lemma]{Theorem}

\theoremstyle{definition}
\newtheorem{remark}[lemma]{Remark}
\newcommand{\MUNCH}[1]{\relax}

\allowdisplaybreaks
\begin{document}
\title{Lorentzian Lie 3-algebras and their Bagger--Lambert moduli space}
\author{Paul de Medeiros, José Figueroa-O'Farrill and Elena Méndez-Escobar}
\address{Maxwell Institute and School of Mathematics, University of
  Edinburgh, UK}
\email{P.deMedeiros@ed.ac.uk, J.M.Figueroa@ed.ac.uk, E.Mendez@ed.ac.uk}
\date{\today}
\begin{abstract}
  We classify Lie 3-algebras possessing an invariant lorentzian inner
  product.  The indecomposable objects are either one-dimensional,
  simple or in one-to-one correspondence with compact real forms of
  metric semisimple Lie algebras.  We analyse the moduli space of
  classical vacua of the Bagger--Lambert theory corresponding to these
  Lie 3-algebras.  We establish a one-to-one correspondence between
  one branch of the moduli space and compact riemannian symmetric
  spaces.  We analyse the asymptotic behaviour of the moduli space and
  identify a large class of models with moduli branches exhibiting the
  desired $N^{3/2}$ behaviour.
\end{abstract}
\maketitle
\tableofcontents

\section{Introduction and motivation}
\label{sec:intro}

The existence of low-energy effective field theory descriptions of
coincident D-branes in string theory in terms of nonabelian gauge
theories has led to a deeper understanding of the dynamics of these
nonperturbative objects.  A similar approximation modelling the dynamics
of multiple M2- and M5-branes in M-theory would also be greatly
beneficial and could lead to a clearer understanding of the fundamental
degrees of freedom here.  However, the construction of interacting
effective field theories which explicitly realise all the symmetries
expected for M-theory branes has proved elusive.

This is perhaps not surprising given the analyses based on
gravitational thermodynamics \cite{KlebanovTseytlinEntropy} and
absorption cross-sections \cite{KlebanovAbsorption} that suggest the
entropy of a large number of $N$ coincident branes should obey a power
law $N^k$, with $k=2,\frac32,3$ for D-, M2-, M5-branes, respectively.
Thus a supersymmetric Yang--Mills description of D-branes involving
$N\times N$ matrices is natural while the appropriate gauge-theoretic
description of multiple branes in M-theory is less clear.  Based on
the foundational work by Bagger and Lambert \cite{BL1} and Gustavsson
\cite{GustavssonAlgM2}, a superconformal field theory in
three-dimensional Minkowski spacetime was constructed in \cite{BL2} as
a model of multiple M2-branes in M-theory.  The theory is maximally
supersymmetric and scale-invariant with an explicit $\fso(8)$
R-symmetry.  This is consistent with the full $\fosp(8|4)$
superconformal symmetry of the near-horizon geometry of M2-branes in
eleven-dimensional supergravity.  The Lagrangian and its equations of
motion nicely encapsulate several other features expected
\cite{SchwarzChernSimons, BasuHarvey} in a low-energy description of
multiple M2-branes.  These encouraging properties have prompted a
great deal of interest in the Bagger--Lambert theory \cite{BL3,
  MukhiBL, BandresCS, BermanAMM, VanRaamsdonkBL, MorozovMM2,
  LambertTong, DistlerBL, GranNilssonBL, Ho3Lie, BergshoeffMM2,
  GP3Lie, GG3Lie, GPkLie, HoM5M2, GMRBL, BRGTV, HIM-M2toD2rev,
  MorozovBLG, MasahitoBL, HIM-M5MM2, Krishnan, JKKKP, LiWang,
  BanerjeeSenBL, LinBL}.

A novel feature of the Bagger--Lambert theory is that it has a local
gauge symmetry which is not based on a Lie algebra, but rather on a Lie
3-algebra.  The analogue of the Lie bracket $[-,-]$ here being the
3-bracket $[-,-,-]$, an alternating trilinear map on a vector space $V$,
which satisfies a natural generalisation of the Jacobi identity
(sometimes referred to as the fundamental identity).  The dynamical
fields in the Bagger--Lambert model are taken to be valued in $V$ and
consist of eight real bosonic scalars and a fermionic spinor in three
dimensions which transforms as a chiral spinor under the $\fso(8)$
R-symmetry.  There is also a non-dynamical gauge field which takes
values in a Lie subalgebra of $\fgl(V)$.  The on-shell closure of the
supersymmetry transformations for these fields follows from the
fundamental identity.

To obtain the correct equations of motion from a Lagrangian that is
invariant under all the aforementioned symmetries seems to require the
Lie 3-algebra to admit an invariant inner product.  In this paper we
will consider only Lie 3-algebras with a nondegenerate inner product.
The signature of this inner product determines the relative signs of the
kinetic terms for the scalar fields in the Bagger--Lambert Lagrangian.
As with ordinary gauge theory, taking this metric to be
positive-definite would avoid potential issues concerning lack of
unitarity in the quantum theory.  The problem is that there are very few
euclidean metric Lie 3-algebras.  Indeed, as shown in \cite{NagykLie}
(see also \cite{GP3Lie, GG3Lie}), they can always be written as the
direct sum of abelian Lie 3-algebras plus multiple copies of the unique
simple euclidean Lie 3-algebra considered by Bagger and Lambert in their
original construction.  Therefore this assumption is too restrictive.

It should be noted that one could still construct the Lagrangian
and gauge-invariant operators for the Bagger--Lambert theory for a
Lie 3-algebra with a degenerate inner product.  However, one could
not obtain from such a Lagrangian the equations of motion for the
fields that couple only to the degenerate components of this
metric.  Nonetheless, in certain cases it is possible to write down
an auxiliary Lagrangian that is decoupled from the one involving
the nondegenerate components of the inner product to account for
the missing equations of motion (see e.g. \cite{LinBL} where the
metric is degenerate in just one direction).  The hope with this
construction would be to evade the no-go theorem noted above and
perhaps find Lie 3-algebras with invariant degenerate metrics
whose nondegenerate components are positive-definite.  A problem
with such Lie 3-algebras is that they do not seem to allow any new
interactions in the Bagger--Lambert Lagrangian.  As shown in
Remark~\ref{re:degip}, the canonical 4-form in \eqref{eq:4-form}
that appears in all the interaction terms in the Lagrangian can
have no `legs' in the degenerate directions.  Hence the Lagrangian
can always be written in terms of the quotient metric Lie
3-algebra corresponding to the nondegenerate directions.  Since
this quotient Lie 3-algebra is nondegenerate, the problem can be
reduced to the nondegenerate case.

Of course, it is by no means guaranteed that the low-energy effective
theory on multiple M2-branes should have a Lagrangian description and
Gran et al \cite{GranNilssonBL} have considered, purely at the level of
its classical equations of motion, the Bagger--Lambert theory for a
class of Lie 3-algebras (defined by a Lie algebra of one dimension
lower) that do not admit a metric.  Following the novel Higgs mechanism
technique introduced by Mukhi and Papageorgakis \cite{MukhiBL}, they are
able to reduce to the correct equations of motion for the supersymmetric
Yang--Mills description of multiple D2-branes with arbitrary gauge
algebra corresponding to the choice of Lie algebra defining the Lie
3-algebra.

An alternative approach considered recently \cite{GMRBL, BRGTV,
  HIM-M2toD2rev} is to investigate the Bagger--Lambert theory for a
class of Lie 3-algebras (defined by a euclidean Lie algebra in two
dimensions lower) admitting an inner product of lorentzian signature.
It is unclear at present whether the Bagger--Lambert theory associated
with such 3-algebras is really unitary at the quantum level, but there
are some encouraging signs noted in the aforementioned references, based
on the specific structure of the interactions in the Bagger--Lambert
model and the way the ghost-like fields that might give rise to
negative-norm states seem to decouple from the physical Hilbert space.
By giving a vacuum expectation value to one of the scalar fields in a
null direction of the Lie 3-algebra and taking the Lie algebra to be
$\fsu(N)$, the correct reduction of the Bagger--Lambert to
supersymmetric Yang--Mills effective Lagrangian for $N$ D2-branes is
obtained in \cite{GMRBL, HIM-M2toD2rev}.  Moreover, in \cite{BRGTV} it
is suggested that the moduli space of this Bagger--Lambert theory has a
branch that corresponds to the moduli space $(\RR^8)^N/S_N$ of $N$
M2-branes in Minkowski spacetime.  This is to be contrasted with the
moduli space of the Bagger--Lambert theory associated with the unique
simple euclidean Lie 3-algebra used in the original construction of
Bagger and Lambert.  The semi-classical moduli space $(\RR^8 \times
\RR^8)/D_{2k}$ for this theory was obtained in \cite{LambertTong} for
$k=1,2$ and in \cite{DistlerBL} for general $k$, where the integer $k$
in the dihedral group $D_{2k}$ of order $4k$ here corresponds to the
value of the quantised level for the Chern-Simons term in the
Bagger--Lambert Lagrangian.  For $k>2$, the M-theoretic interpretation
of the moduli space is still not entirely clear, but is thought to
describe two M2-branes on a so-called M-fold: essentially a $\ZZ_{2k}$
quotient which acts on both the background spacetime and the M2-branes.

Motivated by the recent results in \cite{GMRBL,BRGTV,HIM-M2toD2rev},
and by the tractability of the problem, we set ourselves the task of
classifying the lorentzian metric Lie 3-algebras and studying in
detail the corresponding moduli space of classical vacua.  This paper
contains our results and is organised as follows.  In Section
\ref{sec:metric-3-lie} we study the structure of metric Lie
3-algebras.  The strong analogy with the case of metric Lie algebras
turns out to be very fruitful and after introducing the necessary, yet
standard, algebraic concepts, we are able to classify the lorentzian
3-algebras in Section \ref{sec:lorentzian-3-lie}.  The main result in
this section is Theorem \ref{th:lorentzian}, which says that the
indecomposable ones are either one-dimensional, simple or belong to a
class whose objects are in one-to-one correspondence with the compact
real forms of metric semisimple Lie algebras.  These latter ones are
precisely the class of lorentzian metric Lie 3-algebras already
discovered in \cite{GMRBL,BRGTV,HIM-M2toD2rev}!

In Section \ref{sec:moduli} we investigate the structure of the moduli
space of maximally supersymmetric vacua of the Bagger--Lambert theory
for indecomposable lorentzian Lie 3-algebras.  This adds to and
complements the preliminary analysis in \cite{BRGTV}.  This moduli space
can be described as the quotient by residual gauge symmetries of a
certain variety in the space of constant scalar fields in the
Bagger--Lambert theory.  We perform a detailed calculation of the
residual symmetries of the vacuum which correspond to automorphisms of
the Lie 3-algebra.  To some extent, we are able to factorise the moduli
space and find it to have two distinct branches (which we label
\emph{degenerate} and \emph{nondegenerate}), according to whether one of
the scalar fields in a null direction of the Lie 3-algebra is zero or
not (in agreement with \cite{BRGTV}, whose \emph{abelian} branch
corresponds to our nondegenerate branch).

The nondegenerate branch has the simpler structure and is of the
form
\begin{equation*}
  \eM_{\text{nondeg}} = \RR^{16} \times \RR^{8r}/\fW~,
\end{equation*}
for a compact semisimple Lie algebra $\fg$ of rank $r$ whose Weyl group
is $\fW$.  The second factor can be identified with (the strong coupling
limit of) the classical moduli space of $N{=}8$ super Yang-Mills theory
with gauge algebra $\fg$ in three-dimensional Minkowski space (see
e.g. \cite{Seiberg16}).  The scalars in the two null directions of the
3-algebra, spanning $\RR^{16}$, are related to the extra $\fu(1)$ degree
of freedom and Yang-Mills coupling in the D2-brane reduction described
in \cite{GMRBL}.

The degenerate branch has a much more intricate structure and is defined
by subspaces $\fp \subset \fg$ satisfying $[\fp,\fp]=\fp^\perp$, i.e.,
the Lie bracket on $\fp$ spans the orthogonal complement $\fp^\perp$ of
$\fp \subset \fg$.  We find a large class of such subspaces $\fp$ that
are in one-to-one correspondence with compact riemannian symmetric
spaces, and have maximal dimension $\half(\dim\fg + \rank\fg)$.  After
properly performing the quotient, it turns out the dimension of the
degenerate branch can be very different for different choices of
symmetric spaces.  It is often of larger dimension than the nondegenerate
branch.  It would be very interesting to understand what the M-theoretic
interpretation of this degenerate branch is.

We end by exploring numerically the asymptotic properties of these
branches as one one takes the dimension of the Lie algebra $\fg$ to be
large and we exhibit a large class of models with the infamous $N^{3/2}$
scaling at large $N$.

\section*{Acknowledgments}

It is a pleasure to thank Costis Papageorgakis, Patricia Ritter, Joan
Simón and especially Dmitri Alekseevsky for useful conversations
and/or correspondence.  PdM is supported by a Seggie-Brown
Postdoctoral Fellowship of the School of Mathematics of the University
of Edinburgh.

\section{Metric Lie 3-algebras}
\label{sec:metric-3-lie}

In this section we study the structure of metric Lie 3-algebras.  Lie
$n$-algebras (for $n>2$) were introduced by Filippov \cite{Filippov}
and studied in a number of subsequent papers by Filippov and other
authors.  Metric Lie $n$-algebras (for $n>2$) seem to have been
considered for the first time in \cite{FOPPluecker}, albeit
tangentially.

\subsection{Basic definitions}
\label{sec:basic-definitions}

Let $V$ be a finite-dimensional real vector space.  Recall that a Lie
algebra structure on $V$ is a linear map $[,]: \Lambda^2 V \to V$
obeying the Jacobi identity.  There are many equivalent ways to think
of the Jacobi identity.  One such way is to say that the endomorphisms
$\ad_x$ of $V$ for all $x \in V$, defined by $\ad_x(z) = [x,z]$, are a
derivation over the bracket, or in other words,
\begin{equation}
  \label{eq:Jacobi}
  [x,[y,z]] = [[x,y],z] + [y,[x,z]]~,
\end{equation}
for all $x,y,z\in V$.  This formulation admits a straight-forward
generalisation to $n$-ary brackets.  In this note we will be
interested in the case where $n=3$.  Thus we define a \textbf{Lie
  3-algebra} (also known as a \textbf{Filippov 3-algebra}) structure
on $V$ to be a linear map $\Phi: \Lambda^3 V \to V$, often written simply
as a 3-bracket $\Phi(x,y,y)=[x,y,z]$, such that for all $x,y\in V$,
the endomorphism $\ad_{x,y}: V \to V$,  defined by
\begin{equation}
  \label{eq:varphi}
  \ad_{x,y}(z) = -\ad_{y,x}(z) = [x,y,z]~,
\end{equation}
is a derivation over $\Phi$.  In other words, we demand that for all
$x,y,z_1,z_2,z_3\in V$,
\begin{equation}
  \label{eq:3-Lie}
  [x,y,[z_1,z_2,z_3]] =   [[x,y,z_1], z_2, z_3] + [z_1,[x,y,z_2],z_3]
  + [z_1,z_2,[x,y,z_3]]~,
\end{equation}
which we call the \textbf{3-Jacobi identity}.  An endomorphism
$\delta: V \to V$ is called a \textbf{derivation} of the Lie
3-algebra if for all $x,y,z \in V$,
\begin{equation}
  \label{eq:derivation}
  \delta[x,y,z] = [\delta x, y, z] + [x, \delta y, z] + [x, y, \delta z]~.
\end{equation}
The 3-Jacobi identity says that for all $x,y\in V$, the endomorphism
$\ad_{x,y}$ is a derivation.  Such derivations are said to be
\textbf{inner}.  Derivations form a Lie subalgebra of $\fgl(V)$ of
which the inner derivations are an ideal.

Now recall that a Lie algebra structure on $V$ is said to be
\textbf{metric}, if there is an inner product (i.e., a nondegenerate
symmetric bilinear form) $b \in S^2V^*$ on $V$, often written simply
as $b(x,y) = \left<x,y\right>$, which is invariant under the action of
$\ad_x$ for all $x\in V$; that is,
\begin{equation}
  \label{eq:metric-Lie}
  \left<[x,y],z\right> + \left< y,[x,z] \right>=0~,
\end{equation}
for all $x,y,z\in V$.  In the same spirit, we say that a Lie 3-algebra
structure $(V,\Phi,b)$ is \textbf{metric} if the inner product $b$ is
invariant under the inner derivations $\ad_{x,y}$ for all $x,y\in V$;
that is,
\begin{equation}
  \label{eq:metric-3-Lie}
  \left<[x,y,z],w\right> + \left< z, [x,y,w] \right> = 0~,
\end{equation}
for all $x,y,z,w \in V$.  In other words, the inner derivations
$\ad_{x,y}$ lie in the Lie subalgebra $\fso(V) < \fgl(V)$ preserving
$b$.  Just like a metric Lie algebra possesses a canonical three-form
$\Omega \in \Lambda^3V^*$, given by $\Omega(x,y,z) =
\left<[x,y],z\right>$, a metric Lie 3-algebra possesses a canonical
$4$-form $F \in \Lambda^4 V^*$, defined by
\begin{equation}
  \label{eq:4-form}
  F(x,y,z,w) = \left<[x,y,z],w\right>~.
\end{equation}

Given two metric Lie 3-algebras $(V_1,\Phi_1,b_1)$ and
$(V_2,\Phi_2,b_2)$, we may form their \textbf{orthogonal direct sum}
$(V_1\oplus V_2,\Phi_1\oplus \Phi_2, b_1 \oplus b_2)$, by declaring
that
\begin{equation*}
  [x_1,x_2,y] = 0 \qquad\text{and}\qquad \left<x_1,x_2\right> = 0~,
\end{equation*}
for all $x_i\in V_i$ and all $y\in V_1 \oplus V_2$.  The resulting
object is again a metric Lie 3-algebra.  A metric Lie 3-algebra is
said to be indecomposable if, roughly speaking, it cannot be written
as an orthogonal direct sum of metric Lie 3-algebras $(V_1\oplus
V_2, \Phi_1\oplus \Phi_2, b_1\oplus b_2)$ with $\dim V_i > 0$.

In order to classify the metric Lie 3-algebras, it is clearly enough
to classify the indecomposable ones.  In order to do this we will find
it convenient to introduce some basic 3-algebraic concepts by
analogy with the theory of Lie algebras.  Most of these concepts can
be found in the foundational paper of Filippov \cite{Filippov}.

\subsection{Some structure theory}
\label{sec:structure}

Let $(V,\Phi)$ be a Lie 3-algebra.  A subspace $W \subset V$ is a
\textbf{subalgebra}, written $W < V$, if $[W,W,W]\subset W$, which is a
convenient shorthand notation for the following: for all $x,y,z \in W$,
$[x,y,z] \in W$.  We will this shorthand notation without further
comment in what follows.  A subalgebra $W<V$ (which could be all of $V$)
is said to be \textbf{abelian} if $[W,W,W]=0$.

If $V,W$ are Lie 3-algebras, then a linear map $\phi: V \to W$ is a
\textbf{homomorphism} if
\begin{equation*}
  \phi[x,y,z] = [\phi(x), \phi(y), \phi(z)]~,
\end{equation*}
for all $x,y,z\in V$.  If $\phi$ is also a vector space isomorphism,
we say that it is an \textbf{isomorphism} of Lie 3-algebras.

It is clear that the image of a homomorphism is a subalgebra and we
expect that the kernel ought to be an ideal.  Indeed, if
$x\in\ker\phi$, then $\phi[x,y,z]=0$ for all $y,z\in V$.  This
suggests the following definition: a subspace $I \subset V$ is an
\textbf{ideal}, written $I \lhd V$, if $[I,V,V]\subset I$.

\begin{lemma}
  There is a one-to-one correspondence between ideals and kernels of
  homomorphisms.
\end{lemma}

\begin{proof}
  The kernel of a homomorphism is an ideal, by definition.  (In fact,
  this motivated the definition.)  Conversely, if $I \lhd V$, then
  $V/I$ is a Lie 3-algebra with bracket
  \begin{equation*}
    [x+I, y+I, z+I] = [x,y,z] + I~,
  \end{equation*}
  and the canonical projection $V \to V/I$ is a homomorphism with
  kernel $I$.
\end{proof}

\begin{lemma}
  If $I,J$ are ideals of $V$ then so is their intersection $I\cap J$
  and their linear span $I + J$, defined as the smallest vector
  subspace containing their union $I \cup J$.
\end{lemma}

\begin{proof}
  Since $I\cap J \subset I$, $[I\cap J, V, V] \subset I$ and since
  $I\cap J \subset J$, $[I\cap J, V, V] \subset J$, hence $[I\cap J,
  V, V] \subset I \cap J$.  Similarly, $[I+J,V,V] \subset [I,V,V] +
  [J, V, V] \subset I + J$.
\end{proof}

An ideal $I \lhd V$ is \textbf{minimal} if any other ideal $J \lhd V$
contained in $I$ is either $0$ or $I$.  Dually, an ideal $I \lhd V$ is
\textbf{maximal} if any other ideal $J \lhd V$ containing $I$ is
either $I$ or $V$.

A Lie 3-algebra is \textbf{simple} if it is not one-dimensional and
every ideal $I\lhd V$ is either $0$ or $V$.

\begin{lemma}\label{le:simplequot}
  If $I\lhd V$ is a maximal ideal, then $V/I$ is simple or
  one-dimensional.
\end{lemma}

\begin{proof}
  Let $\pi: V \to V/I$ denote the natural surjection, suppose that
  $J\subset V/I$ is an ideal and let let $\pi^{-1}J = \left\{x \in
    V \middle | \pi(x) \in J\right\}$.  Then $\pi^{-1}J$ is an
  ideal of $V$: $\pi[\pi^{-1}J,V,V] = [J, V/I, V/I] \subset
  J$, whence $[\pi^{-1}J,V,V] \subset \pi^{-1}J$.  Since $I = \ker
  \pi$, $I$ is contained in $\pi^{-1}J$, but since $I$ is maximal
  $\pi^{-1}J=I$ or $\pi^{-1}J=V$.  In the former case, $J =
  \pi\pi^{-1}J = \pi I = 0$ and in the latter $J = \pi\pi^{-1}J = \pi
  V = V/I$.  Hence $V/I$ has no proper ideals.
\end{proof}

Simple Lie $n$-algebras have been classified.  In particular, for $n=3$
we have the following

\begin{theorem}[\cite{LingSimple}]\label{th:simple}
  A simple real Lie 3-algebra is isomorphic to one of the
  four-dimensional Lie 3-algebras defined, relative to a basis
  $\be_i$, by
  \begin{equation}
    \label{eq:simple-3-Lie}
    [\be_i,\be_j,\be_k] = \sum_{\ell=1}^4 \varepsilon_{ijk\ell}
    \lambda_\ell \be_\ell~,
  \end{equation}
  for some $\lambda_\ell$, all nonzero.
\end{theorem}

It is plain to see that simple real Lie 3-algebras admit invariant
metrics of any signature: euclidean, lorentzian or split.  Indeed, the
Lie 3-algebra in \eqref{eq:simple-3-Lie} leaves invariant the diagonal
metric with entries $(1/\lambda_1, 1/\lambda_2, 1/\lambda_3,
1/\lambda_4)$.  One can further change to a basis where the
$\lambda_i$ are signs.  In particular this shows that up to homothety
(i.e., a rescaling of the inner product) there are unique simple metric
Lie 3-algebras with euclidean and lorentzian signatures, corresponding
to choosing $\lambda_i$ to be $(1,1,1,1)$ and $(-1,1,1,1)$,
respectively.  The euclidean case is the original Lie 3-algebra which
was used in Appendix~A of \cite{BL2}.

The image $[V,V,V] \subset V$ of $\Phi: \Lambda^3 V \to V$ is an ideal
called the \textbf{derived ideal} of $V$.  Another ideal is
provided by the \textbf{centre} $Z$, defined by
\begin{equation*}
  Z = \left\{z \in V\middle | [z,x,y]=0,~\forall x,y\in V \right\}~.
\end{equation*}
In other words, $[Z,V,V]=0$.  More generally the \textbf{centraliser}
$Z(W)$ of a subspace $W \subset V$ is defined by
\begin{equation*}
  Z(W) = \left\{z \in V\middle | [z,w,y]=0,~\forall w\in W,y\in V
  \right\}~,
\end{equation*}
or equivalently $[Z(W),W,V]=0$ (thus $Z(V)=Z$).  It follows from the Jacobi identity
\eqref{eq:3-Lie} that $Z(W)$ is a subalgebra.

From now on let $(V,\Phi,b)$ be a metric Lie 3-algebra.  If $W
\subset V$ is any subspace, we define
\begin{equation*}
  W^\perp = \left\{v \in V\middle | \left<v,w\right>=0~,\forall w\in
      W\right\}~.
\end{equation*}
Notice that $(W^\perp)^\perp = W$.  We say that $W$ is
\textbf{nondegenerate}, if $W \cap W^\perp = 0$, whence $V = W \oplus
W^\perp$; \textbf{isotropic}, if $W \subset W^\perp$; and
\textbf{coisotropic}, if $W \supset W^\perp$.  Of course, in
positive-definite signature, all subspaces are nondegenerate.  A
metric Lie 3-algebra is said to be \textbf{indecomposable} if it is
not isomorphic to a direct sum of orthogonal ideals or, equivalently,
if it does not possess any proper nondegenerate ideals: for if $I\lhd
V$ is nondegenerate, $V = I \oplus I^\perp$ is an orthogonal direct
sum of ideals.

The proof of the following lemma is routine.

\begin{lemma}\label{le:coisoquot}
  Let $I\lhd V$ be a coisotropic ideal of a metric Lie 3-algebra.
  Then $I/I^\perp$ is a metric Lie 3-algebra.
\end{lemma}


\begin{lemma}
  Let $V$ be a metric Lie 3-algebra.  Then the centre is the
  orthogonal subspace to the derived ideal; that is,
  $[V,V,V]=Z^\perp$.
\end{lemma}

\begin{proof}
  Let $z\in Z$, then for all $x,y,w\in V$, $0 = \left<[x,y,z],w\right>
  = - \left<[x,y,w],z\right>$, whence $z \in [V,V,V]^\perp$ and
  $Z\subset [V,V,V]^\perp$.  Conversely, let $z \in [V,V,V]^\perp$.
  This means that for all $x,y,w\in V$,
  \begin{equation*}
    0 = \left<z,[x,y,w]\right> = - \left<[x,y,z],w\right>~,
  \end{equation*}
  which implies that $[x,y,z]=0$ for all $x,y\in V$ and hence $z \in Z$.
  In other words, $[V,V,V]^\perp \subset Z$.
\end{proof}

\begin{proposition}\label{pr:ideals}
  Let $V$ be a metric Lie 3-algebra and $I \lhd V$ be an ideal.
  Then
  \begin{enumerate}
  \item $I^\perp \lhd V$ is also an ideal;
  \item $I^\perp\lhd Z(I)$; and
  \item if $I$ is minimal then $I^\perp$ is maximal.
  \end{enumerate}
\end{proposition}

\begin{proof}
  \begin{enumerate}
  \item For all $x,y\in V$, $u\in I$ and $v\in I^\perp$,
    $\left<[v,x,y],u\right> = - \left<[x,y,u],v\right> = 0$, since
    $[x,y,u]\in I$.  Therefore $[v,x,y] \in I^\perp$.
  \item For all $u \in I^\perp$, $v\in I$ and $x,y\in V$, consider
    $\left<[u,v,x],y\right> = - \left<[x,y,v],u\right> =0$ since $I$
    is an ideal, which means that $[u,v,x]=0$, whence
    $[I,I^\perp,V]=0$.
  \item Let $J \supset I^\perp$ be an ideal.  Taking perpendiculars,
    $J^\perp \subset I$.  Since $I$ is minimal, $J^\perp = 0$ or
    $J^\perp = I$, whence $J = V$ or $J= I^\perp$ and $I^\perp$ is
    maximal.
  \end{enumerate}
\end{proof}

\begin{remark}\label{re:degip}
  Although we have been assuming that the inner product is
  nondegenerate, let us make a remark concerning the possibility of a
  degenerate inner product.  Let $V^\perp$ denote the radical of the
  inner product; that is $x \in V^\perp$ if $\left<x,y\right>=0$ for all
  $y\in V$.  It follows immediately by invariance of the inner product
  that $V^\perp \lhd V$ is an ideal.  This means that if $F$ is the
  4-form in \eqref{eq:4-form}, then $F(x,-,-,-)=0$ for all $x\in
  V^\perp$.  In other words, the 4-form is the pull-back of the 4-form
  on the quotient metric Lie 3-algebra $V/V^\perp$.  This means that the
  degrees of freedom corresponding to $V^\perp$ seem to effectively
  decouple, yielding a Bagger--Lambert Lagrangian with Lie 3-algebra
  $V/V^\perp$.  Of course, the symmetry transformations themselves (that
  do not involve the metric) will generally mix up the degrees of
  freedom on $V/V^\perp$ and $V^\perp$, though clearly the truncation to
  $V/V^\perp$ is consistent.
\end{remark}

\subsection{Structure of metric Lie 3-algebras}
\label{sec:structure-metric-3}

We now investigate the structure of metric Lie 3-algebras.  As in
the case of Lie algebras \cite{MedinaRevoy,FSsug,FSalgebra}, there is
a subtle interplay between ideals and the inner product.

If a Lie 3-algebra is not simple or one-dimensional, then it has a
proper ideal and hence a minimal ideal.  Let $I\lhd V$ be a minimal
ideal of a metric Lie 3-algebra.  Then $I \cap I^\perp$, being an
ideal contained in $I$, is either $0$ or $I$.  In other words, minimal
ideals are either nondegenerate or isotropic.  If nondegenerate, $V = 
I \oplus I^\perp$ is decomposable.  Therefore if $V$ is
indecomposable, $I$ is isotropic.  Moreover, by
Proposition~\ref{pr:ideals} (2), $I$ is abelian and furthermore, because
$I$ is isotropic, $[I,I,V]=0$.

It follows that if $V$ is euclidean and indecomposable, it is either
one-dimensional or simple, whence of the form \eqref{eq:simple-3-Lie}
with all $\lambda_i$ positive.  As we will see below, one can choose
an orthogonal (but not orthonormal) basis for $V$ where the
$\lambda_i$ are equal to $1$.  This result, originally due to
\cite{NagykLie}, was conjectured in \cite{FOPPluecker}, both of which
also treat the case of Lie $n$-algebras for $n>3$.  This result has been
rediscovered more recently for Lie 3-algebras in \cite{GP3Lie,GG3Lie}
and for $n>3$ in \cite{GPkLie}.  Here it is seen to follow structurally
as a corollary of Theorem~\ref{th:simple}.  The result for $n>3$ also
follows structurally from the $n>3$ version of that theorem along the
same lines.

Let $V$ be an indecomposable metric Lie 3-algebra.  Then $V$ is
either simple, one-dimensional, or possesses a proper minimal ideal
$I$ which is isotropic and obeys $[I,I,V]=0$.  The perpendicular ideal
$I^\perp$ is maximal and hence by Lemma~\ref{le:simplequot}, $U :=
V/I^\perp$ is simple or one-dimensional, whereas by
Lemma~\ref{le:coisoquot}, $W:=I^\perp/I$ is a metric Lie 3-algebra.

The inner product on $V$ induces a nondegenerate pairing $g: U \otimes
I \to \RR$.  Indeed, let $[u] = u + I^\perp \in U$ and $v\in I$.  Then
we define $g([u],v) = \left<u,v\right>$, which is clearly independent
of the coset representative for $[u]$.  In particular, $I \cong U^*$
is either one- or four-dimensional.  If the signature of the metric of
$W$ is $(p,q)$, that of $V$ is $(p+k,q+k)$ where $k = \dim I = \dim
U$.  So that if $V$ is to have lorentzian signature, $k = 1$ and $W$
must be euclidean; although not necessarily indecomposable.

In the next section we will classify indecomposable lorentzian Lie
3-algebras.  The technique is analogous to the classification of
indecomposable metric Lie algebras given in \cite{MedinaRevoy} and its
refinement in \cite{FSsug,FSalgebra}.  The lorentzian Lie algebras
have been classified in \cite{MedinaLorentzian} (see also
\cite[§2.3]{FSPL}).  By the same techniques it is possible
\cite{2p3Lie} to classify metric Lie 3-algebras with signature $(2,*)$
and to prove a structure theorem for the case of general signature.
Similarly it is possible to classify lorentzian Lie $n$-algebras for
$n>3$ \cite{JMFLorNLie}.

\section{Lorentzian Lie 3-algebras}
\label{sec:lorentzian-3-lie}

A lorentzian Lie 3-algebra decomposes into one lorentzian
indecomposable factor and zero or more indecomposable euclidean
factors.  As discussed above, the indecomposable euclidean Lie
3-algebras are either one-dimensional or simple.  On the other hand, an
indecomposable lorentzian Lie 3-algebra is either one-dimensional,
simple or else possesses a one-dimensional isotropic minimal ideal.
It is this latter case which remains to be treated and we do so now.

The quotient Lie 3-algebra $U=V/I^\perp$ is also one-dimensional.
Let $u \in V$ be such that $u \not\in I^\perp$, whence its image in
$U$ generates it.  Because $I \cong U^*$, there is $v \in I$ such that
$\left<u,v\right> = 1$.  Complete it to a basis $(v,x_a)$ for
$I^\perp$.  Then $(u,v,x_a)$ is a basis for $V$, with $(x_a)$ spanning
a subspace isomorphic to $W=I^\perp/I$ and which, with a slight abuse
of notation, we will also denote $W$.  It is possible to choose $u$ so
that $\left<u,u\right> = 0$ and such that $\left<u,x\right>=0$ for all
$x\in W$.  Indeed, given any $u$, the map $x \mapsto \left<u,x\right>$
defines an element in the dual $W^*$.  Since the restriction of the
inner product to $W$ is nondegenerate, there is some $z\in W$ such
that $\left<u,x\right> = \left<z,x\right>$ for all $x \in W$.  We let
$u' = u - z$.  This still obeys $\left<u',v\right>=1$ and now also
$\left<u',x\right>=0$ for all $x\in W$.  Finally let $u'' = u' - \half
\left<u',u'\right> v$, which still satisfies $\left<u'',v\right> = 1$,
$\left<u'',x\right>=0$ for all $x\in W$, but now satisfies
$\left<u'',u''\right>=0$ as well.

In this basis, there are four kinds of 3-brackets: $[u,v,x]$,
$[u,x,y]$, $[v,x,y]$ and $[x,y,z]$ where $x,y,z\in W$.  From 
Proposition~\ref{pr:ideals} (2), it is immediate that
$[u,v,x]=0=[v,x,y]$, whence $v$ is central.  In summary, the only
nonzero 3-brackets are, using the summation convention,
\begin{equation}
  \label{eq:3brackets}
  [u,x_a,x_b] = f_{ab}{}^c x_c \qquad\text{and}\qquad
  [x_a,x_b,x_c] = - f_{abc} v + \phi_{abc}{}^d x_d~,
\end{equation}
where $f_{abc} = \left<[u,x_a,x_b],x_c\right>$. (Notice that an
additional $\omega_{ab} v$ term that might have occurred on the right hand
side of the first 3-bracket must vanish by taking the inner product with
$u$.)  The 3-Jacobi identity is equivalent to the following two
conditions:
\begin{enumerate}
\item $[x_a,x_b] := f_{ab}{}^c x_c$ defines a Lie algebra structure on
  $W$, which leaves the inner product invariant due to the
  skewsymmetry of $f_{abc} = \left<[x_a,x_b],x_c\right>$; and
\item $[x_a,x_b,x_c]_W := \phi_{abc}{}^d x_d$ defines a euclidean
  Lie 3-algebra structure on $W$ which is ad-invariant with respect
  to the Lie algebra structure.
\end{enumerate}

We will show below that for $V$ indecomposable, $\phi\equiv 0$ so that
the Lie 3-algebra structure on $W$ is abelian, but not before
discussing a family of Lie algebras associated to every Lie
3-algebra.

\subsection{A family of metric Lie algebras}
\label{sec:family-lie}

Let $(V,\Phi)$ be a Lie 3-algebra.  It was already observed in
\cite{Filippov} that every $z\in V$ defines a bracket $[-,-]_z:
\Lambda^2V \to V$ by
\begin{equation}
  \label{eq:Lie-z}
  [x,y]_z := [x,y,z]~,
\end{equation}
which obeys the Jacobi identity as a consequence of the 3-Jacobi
identity \eqref{eq:3-Lie}.  Thus $[-,-]_z$ defines on $V$ a Lie
algebra structure for which $z$ is a central element.  In other words,
a Lie 3-algebra $V$ defines a family of Lie algebras on $V$
parametrised linearly by $V$ itself.  Letting $\eL_V \subset
\Lambda^2V^* \otimes V$ denote the space of Lie algebra structures on
$V$, a Lie 3-algebra structure on $V$ defines a linear embedding $V
\into \Lambda^2V^*\otimes V$ whose image lies in $\eL_V$.  Although it
would be tempting to characterise Lie 3-algebras in this way, it is
known \cite{BremnerNLie,RotkiewiczNLie} however that this condition is
strictly weaker than the 3-Jacobi identity.  It is not known whether
this is still the case for metric Lie 3-algebras.

If $(V,\Phi,b)$ is a metric Lie 3-algebra, then each of the Lie algebras
$(V,[-,-]_z,b)$ is a metric Lie algebra.  Let $V$ be a simple euclidean
Lie 3-algebra.  It is possible to change to a basis
$(\be_1,\dots,\be_4)$ where the 3-bracket is
\begin{equation}
  \label{eq:simple3Lienf}
  [\be_i,\be_j,\be_k] = \varepsilon_{ijk\ell} \be_\ell~,
\end{equation}
using the summation convention.  Moreover, such a basis is orthogonal,
but not necessarily orthonormal.  Thus there is a one parameter family
of such metric Lie 3-algebras, distinguished by the scale of the inner
product.  We will denote the simple Lie 3-algebra with the above
3-brackets by $\fs$.  Fixing any nonzero $x \in \fs$, the Lie algebra
$[-,-]_x$ is isomorphic to $\fso(3) \oplus \RR$, where the $\fso(3)$ is
the orthogonal Lie algebra in the perpendicular complement of the line
containing $x$.  Under the adjoint action of this Lie algebra, the
vector space $\fs$ decomposes into $\fs = \RR x \oplus x^\perp$.

\subsection{Indecomposable lorentzian Lie 3-algebras}
\label{sec:indec-lor}

We are now ready to classify the indecomposable lorentzian Lie
3-algebras.  We have previously shown that such an algebra is given in a
basis $(u,v,x_a)$ by the 3-bracket in \eqref{eq:3brackets}.  We will
now show that if $V$ is indecomposable, then $\phi$ necessarily
vanishes.

The tensor $\phi_{abc}{}^d$ defines a euclidean Lie 3-algebra
structure on $W$.  The most general euclidean Lie 3-algebra is an
orthogonal direct sum $W = \fa \oplus \fs_1 \oplus \dots \oplus
\fs_m$, where $\fa$ is an $n$-dimensional abelian Lie 3-algebra and
the $\fs_i$ are $m$ copies of the simple Lie 3-algebra with 3-brackets
given by \eqref{eq:simple3Lienf}.  The inner product is such that the
above direct sums are orthogonal, and the inner products on each of
the factors is positive-definite.

The \emph{Lie algebra} structure on $W$ is such that its adjoint
representation preserves both the 3-brackets and the inner product,
whence $\ad W$ is contained in $\fso(\fa) \oplus \fso(\fs_1) \oplus
\dots \oplus \fso(\fs_m)$.  Indeed, for any $x\in W$, $\ad_x$
preserves the Lie 3-bracket, whence also the ``volume'' forms on each
of the simple factors.  In turn this means that $\ad_x$ preserves the
subspaces $\fs$ themselves.  To see this, let
$(\be_1,\be_2,\be_3,\be_4)$ be a basis for one of the simple factors,
say $\fs_1$, and let $\be_1\wedge \be_2 \wedge \be_3 \wedge \be_4$ be
the corresponding volume form.  Invariance under $\ad_x$ means
\begin{equation*}
  [x,\be_1] \wedge \be_2 \wedge \be_3 \wedge \be_4 + 
  \be_1 \wedge [x,\be_2] \wedge \be_3 \wedge \be_4 + 
  \be_1 \wedge \be_2 \wedge [x,\be_3] \wedge \be_4 + 
  \be_1 \wedge \be_2 \wedge \be_3 \wedge [x,\be_4] = 0~.
\end{equation*}
Now by invariance of the inner product, $[x,\be_i] \perp \be_i$,
whence we may write it as $[x,\be_i] = y_i + z_i$, where $y_i \in
\fs_1 \cap \be_i^\perp$ and $z_i \in \fs_1^\perp$.  Back into the
above equation,
\begin{equation*}
  z_1 \wedge \be_2 \wedge \be_3 \wedge \be_4 + 
  \be_1 \wedge z_2 \wedge \be_3 \wedge \be_4 + 
  \be_1 \wedge \be_2 \wedge z_3 \wedge \be_4 + 
  \be_1 \wedge \be_2 \wedge \be_3 \wedge z_4 = 0~.
\end{equation*}
Each of the above four terms is linearly independent, whence $z_i=0$
and $\ad_x$ indeed preserves $\fs_1$.  This means that each simple
factor is a submodule of the adjoint representation and, hence that so
is their direct sum.  Finally, by invariance of the inner product, so
is its perpendicular complement $\fa$.  In other words, the adjoint
representation is contained in $\fso(\fa) \oplus \fso(\fs_1) \oplus
\dots \oplus \fso(\fs_m)$.

This decomposition of the adjoint representation now implies a
decomposition of the Lie algebra itself as $W = \fg \oplus \fh_1
\oplus \dots \oplus \fh_m$, where $\fg$ is an $n$-dimensional
euclidean Lie algebra (i.e., with $\ad \fg < \fso(\fa)$) and each
$\fh_i$ is a four-dimensional euclidean Lie algebra (i.e., $\ad\fh_i <
\fso(\fs)$).  Indeed, if $x$ and $y$ belong to different orthogonal
summands of the vector space $W$, then $[x,y]$ belongs to the same
summand as $y$ when understood as $\ad_x(y)$ and to the same summand
as $x$ when understood as $\ad_y(x)$.  Since these summands are
orthogonal, $[x,y]=0$.

Now euclidean Lie algebras are reductive; that is, a direct sum of a
compact semisimple Lie algebra and an abelian Lie algebra.  By
inspection there are precisely two isomorphism classes of
four-dimensional euclidean Lie algebras: the abelian $4$-dimensional
Lie algebra $\RR^4$ and $\fso(3) \oplus \RR$.  Hence $\fso(\fs)$ has
to be isomorphic to one of those.

We will now show that every $\fs$ summand in $W$ factorises in $V$,
contradicting the assumption that $V$ is indecomposable.

Consider one such $\fs$ summand, say $\fs_1$.  The corresponding Lie
algebra $\fh_1$ is either abelian or isomorphic to $\fso(3) \oplus
\RR$.  If $\fh_1$ is abelian, so that the structure constants
vanish, then for any $x \in \fs_1$, $[u,x,V]=0$ and $[x,y,V]=0$ for
any $y \in W$ perpendicular to $\fs_1$ .  Hence $\fs_1\lhd V$ is a
nondegenerate ideal, contradicting the indecomposability of $V$.

If $\fh_1 \cong \fso(3) \oplus \RR$, its adjoint algebra $\ad\fh_1$ is
an $\fso(3)$ subalgebra of $\fso(\fs_1) \cong \fso(4)$, which
therefore leaves a line $\ell \subset \fs_1$ invariant.  The Lie
algebra structure on $\fh_1$ thus coincides with that given by the Lie
bracket $[-,-]_x = [x,-,-]_W$, for some $x\in \ell$, induced from
the Lie 3-algebra structure on $\fs_1$.  In other words, $[u,y,z] =
[y,z]_x = [x,y,z]_W$ for all $y,z\in W$.  This allows us to ``twist''
$\fs_1$ into a nondegenerate ideal of $V$.   Indeed, define now
\begin{equation}
  \label{eq:newcoords}
  u' = u - x - \half |x|^2 v \qquad\text{and}\qquad
  y' = y + \left<y,x\right> v~,
\end{equation}
for all $y \in \fs_1$.  Then $[u',y',z']= 0$ for all $y,z\in \fs_1$,
and, using that $v$ is central,
\begin{align*}
  [y',z',w'] = [y,z,w] &= - \left<[y,z],w\right> v + [y,z,w]_W\\
  &= - \left<[x,y,z]_W,w\right> v + [y,z,w]_W \\
  &=  \left<[y,z,w]_W,x\right> v + [y,z,w]_W\\
  &= [y,z,w]'_W~.
\end{align*}
Moreover, for every $y\in \fs_1$,
\begin{equation*}
  \left<u', y'\right> = \left<u - x- \half |x|^2 v, y +
    \left<x,y\right> v\right> = \left<x,y\right> \left<u,v\right> -
  \left<x,y\right> = 0~,
\end{equation*}
and finally
\begin{equation*}
  \left<u',u'\right> = \left<u-x-\half |x|^2 v, u-x-\half |x|^2
    v\right> = - |x|^2 \left<u,v\right> + \left<x,x\right> = 0~.
\end{equation*}
In other words, the subspace of $V$ spanned by the $y'$ for
$y\in\fs_1$ is a nondegenerate ideal of $V$, contradicting again the
fact that $V$ is indecomposable.

Consequently there can be no $\fs$'s in $W$, whence as a Lie
3-algebra, $W$ is abelian.  As a Lie algebra it is euclidean, whence
reductive.  However the abelian summand commutes with $u$, hence it is
central in $V$, again contradicting the fact that it is
indecomposable.  Therefore as a Lie algebra $W$ is compact semisimple.

In summary, we have proved the following

\begin{theorem}\label{th:lorentzian}
  Let $(V,\Phi,b)$ be an indecomposable lorentzian Lie 3-algebra.
  Then it is either one-dimensional, simple, or else there is a Witt
  basis $(u,v,x_a)$, with $u,v$ complementary null directions, such
  that the nonzero 3-brackets take the form
  \begin{equation*}
    [u,x_a,x_b] = f_{ab}{}^c x_c \qquad\text{and}\qquad
    [x_a,x_b,x_c] = - f_{abc} v~,
  \end{equation*}
  where $[x_a,x_b] = f_{ab}{}^c x_c$ makes the span of the $(x_a)$
  into a compact semisimple Lie algebra and $f_{abc} =
  \left<[x_a,x_b],x_c\right>$.
\end{theorem}

These latter Lie 3-algebras have been discovered independently in
\cite{GMRBL,BRGTV,HIM-M2toD2rev}, albeit in some cases in a slightly
different form.  It should be remarked that they provide explicit
counterexamples to the lorentzian conjecture of \cite{FOPPluecker},
simply by taking the semisimple Lie algebra to be anything but a
direct product of $\fso(3)$'s.  Since the main focus in
\cite{FOPPluecker} was on middle-dimensional forms in low dimension,
such examples did not arise.

Paraphrasing the theorem, the class of indecomposable lorentzian Lie
3-algebras are in one-to-one correspondence with the class of euclidean
metric semisimple Lie algebras, by which we mean a compact semisimple
Lie algebra \emph{and} a choice of invariant inner product.  This choice
involves a choice of scale for each simple factor.

A final remark is that the classification of indecomposable lorentzian
Lie 3-algebras is analogous to the classification of indecomposable
lorentzian Lie algebras, which as shown in \cite{MedinaLorentzian}
(see also \cite[§2.3]{FSPL}) are either one-dimensional, simple, or
obtained as a double extension \cite{MedinaRevoy,FSsug,FSalgebra} of
an abelian euclidean Lie algebra $\fg$ by a one-dimensional Lie
algebra acting on $\fg$ via a skew-symmetric endomorphism.  In the Lie
3-algebra case, we have an analogous result, with the action of the
endomorphism being replaced by a semisimple Lie algebra.

\section{The Bagger--Lambert moduli space}
\label{sec:moduli}

\subsection{Basic definitions}
\label{sec:basic-moduli}

The space of (maximally supersymmetric) \textbf{classical vacua} of the
Bagger--Lambert theory associated to a Lie 3-algebra $V$ is defined as
follows:
\begin{equation*}
  \eV = \left\{\phi \in \Hom(\RR^8,V) \middle | [\phi(\bx),\phi(\by),
    \phi(\bz)] = 0~\forall\bx,\by,\bz\in\RR^8\right\}~.
\end{equation*}
In other words, a linear map $\phi: \RR^8 \to V$ belongs to $\eV$
if and only if its image lies in an abelian subalgebra of $V$.
This assumption guarantees that all the Bagger--Lambert
supersymmetry transformations vanish if one sets the gauge field
and fermions to zero and the scalars equal to the constants
$\phi$.  If $A<V$ is an abelian subalgebra, then let us define
$\eV_A := \Hom(\RR^8,A)$, whence
\begin{equation*}
  \eV = \bigcup_{\substack{A<V\\\text{abelian}}} \eV_A~.
\end{equation*}
If $A,B$ are abelian subalgebras of $V$ with $A<B$, then $\eV_A
\subset \eV_B$, whence we can write $\eV$ as a union
\begin{equation*}
  \eV = \bigcup_{\substack{A<V\\\text{maximal abelian}}} \eV_A
\end{equation*}
of maximal subspaces, in the sense that no two subspaces appearing in
the above sum are contained in one another.  We see that $\eV$ is
therefore given by the set union of linear subspaces in $\Hom(\RR^8,V)$,
parametrised by the set of maximal abelian subalgebras of $V$.  In other
words, a maximal abelian subalgebra of $V$ determines a ``branch'' of
the classical space of vacua.  Some of these branches will be
gauge-related, hence the need to quotient by gauge transformations.
Since we have chosen a gauge in which the gauge field vanishes, we are
only allowed to quotient by gauge transformations which preserve this
choice of gauge.  In particular they are constant, whence they define an
(invariant) subgroup of the automorphisms of the Lie 3-algebra.

We define the \textbf{automorphism group} $\Aut V$ of the Lie
3-algebra $V$ to be the subgroup of $\GL(V)$ which preserves the
3-bracket; that is,
\begin{equation*}
  \Aut V = \left\{ g \in \GL(V) \middle | g[x,y,z] = [gx, gy,
    gz],~\forall x,y,z\in V\right\}~.
\end{equation*}
Its Lie algebra $\Der V$ is the Lie subalgebra of $\fgl(V)$ consisting
of derivations of the 3-bracket.  Let $\ad V = \left\{ \ad_{x,y} \middle
  | x,y \in V\right\}$ denote the Lie subalgebra of $\Der V$ consisting
of inner derivations.  It is the Lie algebra of a normal subgroup $\Ad V
\lhd \Aut V$, which we will call the group of \textbf{inner
  automorphisms}; although we should keep in mind that the nomenclature
is somewhat misleading in the absence of a notion of Lie 3-group.  The
group $\Ad V$ is the closest thing one has to a gauge group in the
Bagger--Lambert theory.

$\Aut V$, and hence $\Ad V$, act on $\Hom(\RR^8,V)$ by ignoring the
$\RR^8$ and acting on $V$ via the defining representation.  In other
words, if $g\in\Aut V$ and $\phi:\RR^8 \to V$, then $(g \cdot\phi)(\bx)
= g \phi(\bx)$, for all $\bx \in \RR^8$.  It is clear that if $\phi \in
\eV$, then $g \cdot \phi \in \eV$, whence we may define the (classical)
\textbf{moduli space} of the Bagger--Lambert theory associated with $V$
as the gauge-equivalence classes of vacuum configurations in $\eV$,
where two vacuum configurations are gauge-equivalent if they related by
the action of $\Ad V$.

\subsection{Automorphisms}
\label{sec:autos}

We will now study the automorphisms $\Aut V$ of the indecomposable
lorentzian Lie 3-algebra $V$ built out of a euclidean semisimple Lie
algebra $\fg$ as described in Theorem \ref{th:lorentzian}, paying close
attention to the inner automorphisms $\Ad V$.

As seen above, $V = \RR u \oplus \RR v \oplus \fg$ as a vector space,
with inner product given by extending the inner product on $\fg$ to $V$
in such a way that $u,v$ are orthogonal to $\fg$ and obey
$\left<u,u\right>= 0 = \left<v,v\right>$ and $\left<u,v\right>=1$, and
where the nonzero 3-brackets are given by
\begin{equation*}
  [u,x,y] = [x,y] \qquad\text{and}\qquad [x,y,z]= -
  \left<[x,y],z\right> v~,
\end{equation*}
for all $x,y,z \in \fg$.

Let $\varphi \in \Aut V$.  It follows that $\varphi$ preserves the
centre $Z \lhd V$.  Indeed, if $z \in Z$, then for all $x,y\in V$,
$[\varphi(z), \varphi(x), \varphi(y)] = \varphi[z,x,y] = 0$.  This
means that $[\varphi(z),\varphi(V),\varphi(V)]=0$, but
$\varphi(V)=V$ since $\varphi$ is vector space isomorphism, whence
$\varphi(z) \in Z$.  Similarly, $\varphi$ preserves the derived
ideal $[V,V,V]$.  Indeed, $\varphi[V,V,V] = [\varphi(V),
\varphi(V), \varphi(V)] = [V,V,V]$.  This means that $\varphi$
must take the following form:
\begin{align*}
  \varphi(v) &= \alpha v\\
  \varphi(u) &= \beta u + \gamma v + t\\
  \varphi(x) &= f(x) + \left<w,x\right> v~,
\end{align*}
for all $x \in \fg$ and where $\alpha,\beta,\gamma \in \RR$, $w,t \in
\fg$ and $f: \fg \to \fg$.  Invertibility of $\varphi$ forces
$\alpha,\beta$ to be nonzero and $f$ to be invertible.

We will now determine the most general $\varphi$ preserving the
3-brackets.  For all $x,y,z\in \fg$,
\begin{equation*}
  \varphi[x,y,z] = - \left<[x,y],z\right> \varphi(v) = - \alpha
  \left<[x,y],z\right> v~,
\end{equation*}
but also
\begin{equation*}
  \varphi[x,y,z] = [\varphi(x),\varphi(y),\varphi(z)] = [f(x),f(y),f(z)]
  = - \left<[f(x),f(y)],f(z)\right> v~,
\end{equation*}
whence we arrive at
\begin{equation}
  \label{eq:auto1}
  \left<[f(x),f(y)],f(z)\right> = \alpha \left<[x,y],z\right>~.
\end{equation}
For all $x,y\in \fg$,
\begin{equation*}
  \varphi[u,x,y] = \varphi[x,y] = f[x,y] + \left<w,[x,y]\right> v~,
\end{equation*}
but also
\begin{equation*}
  \varphi[u,x,y] = [\varphi(u),\varphi(x),\varphi(y)] = [\beta u + t,
  f(x),f(y)] = \beta [f(x),f(y)] - \left<t,[f(x),f(y)]\right> v~,
\end{equation*}
which yields two equations
\begin{equation}
  \label{eq:auto2}
  \beta [f(x),f(y)] = f[x,y]~,
\end{equation}
and
\begin{equation}
  \label{eq:auto3}
  \left<w,[x,y]\right> = - \left<t,[f(x),f(y)]\right>~.
\end{equation}
Multiplying equation \eqref{eq:auto1} by $\beta$ (which is nonzero) and
using equation \eqref{eq:auto2}, we arrive at
\begin{equation*}
  \left<f[x,y],f(z)\right> = \alpha\beta \left<[x,y],z\right>~,
\end{equation*}
since $\fg$ is semisimple, $[\fg,\fg]=\fg$, whence this is equivalent to
\begin{equation}
  \label{eq:auto4}
  \left<f(x), f(y)\right> = \alpha\beta \left<x,y\right>~,
\end{equation}
for all $x,y\in\fg$.
Now we multiply equation \eqref{eq:auto3} by $\beta$ and again use
equation \eqref{eq:auto2} to obtain
\begin{equation*}
    \beta \left<w,[x,y]\right> = - \left<t,f[x,y]\right> = - \left<f^*t,
    [x,y]\right>~,
\end{equation*}
and again using that $[\fg,\fg]=\fg$, we see that $w = - \beta^{-1} f^*t$,
where $f^*$ is the adjoint of $f$ relative to the inner product.

Let us define
\begin{equation*}
  \Aut_\lambda\fg = \left\{f : \fg \to \fg~\text{invertible}~ \middle |
    \lambda[f(x),f(y)] = f[x,y]\right\}~.
\end{equation*}
For $\lambda =1$ we have the automorphism group of $\fg$.
The proof of the following lemma is routine.

\begin{lemma}
  If $f \in \Aut_\lambda\fg$ then $f^{-1} \in \Aut_{1/\lambda}\fg$.  If
  in addition $g \in \Aut_\mu\fg$ then $f \circ g \in
  \Aut_{\lambda\mu}\fg$.
\end{lemma}

Equation \eqref{eq:auto2} says that $f \in \Aut_\beta\fg$.  Now if
$f_1,f_2 \in \Aut_\beta\fg$, the lemma says that $f_1^{-1} \circ f_2
\in \Aut\fg$, whence any two elements of $\Aut_\beta\fg$ are related
by composition with an automorphism.  Now the map $x \mapsto \beta^{-1} x$
is invertible and belongs to $\Aut_\beta\fg$.  Hence the most general
solution to equation \eqref{eq:auto2} is given by $f(x) = \beta^{-1}
a(x)$ for some Lie algebra automorphism $a\in\Aut\fg$.  Substituting
this into equation \eqref{eq:auto4}, we find
\begin{equation}
  \label{eq:auto5}
  \left<a(x),a(y)\right> = \alpha\beta^3\left<x,y\right>~.
\end{equation}
Since the inner product is positive definite, this means that
$\alpha\beta^3 > 0$.

Now decompose $\fg= \fg_1 \oplus \dots \oplus \fg_N$ into simple
factors.  On each of the factors, the inner product is a (negative)
multiple of the Killing form, which is preserved by automorphisms.
Therefore if $a \in \Aut\fg$ preserves each of the factors (this is
the case, e.g., if no two factors are isomorphic) then it preserves the
inner product and we see that $\alpha\beta^3 = 1$.  If $a$ does not
preserve the factors, it permutes them as well as acting by
automorphisms of each of the factors.  However, since $N$ is finite,
some power of $a$ will again preserve the factors and hence some power
of $\alpha\beta^3$ must be equal to $1$, but since $\alpha\beta^3>0$ we
again conclude that $\alpha\beta^3=1$, whence equation \eqref{eq:auto5}
says that $a$ is an isometry.  Let us denote by $\Aut^0\fg$ the
subgroup of automorphisms which are also isometries.  In summary, we
have proved the following

\begin{proposition}
  Every 3-algebra automorphism $\varphi \in \Aut V$ is given by
  \begin{align*}
    \varphi(v) &= \beta^{-3} v\\
    \varphi(u) &= \beta u + \gamma v + t\\
    \varphi(x) &= \beta^{-1} a(x) - \beta^{-2} \left<t,a(x)\right> v~,
  \end{align*}
  for all $x\in\fg$ and where $\beta\in\RR^\times$, $\gamma\in\RR$,
  $t\in\fg$ and $a\in\Aut^0\fg$.
\end{proposition}

Let $\Aut^0V$ denote the subgroup of 3-algebra automorphisms of $V$
which also preserve the inner product.  It is easy to determine such
automorphisms.

\begin{proposition}\label{pr:Aut0V}
  Every 3-algebra automorphism $\varphi \in \Aut^0 V$ preserving the
  inner product is given by
  \begin{align*}
    \varphi(v) &= v\\
    \varphi(u) &= u -\half |t|^2 v + t\\
    \varphi(x) &= a(x) - \left<t,a(x)\right> v~,
  \end{align*}
  for all $x\in\fg$ and where $t\in\fg$ and $a\in\Aut^0\fg$.
\end{proposition}

\begin{proof}
  The condition $\left<\varphi(u),\varphi(v)\right> = 1$ fixes $\beta
  =1$.  The condition $\left<\varphi(u),\varphi(u)\right>=0$ fixes
  $\gamma = - \half |t|^2$.  The rest of the conditions are satisfied
  identically.
\end{proof}

The automorphism generated by $\beta$ (that does not preserve the
inner product) can be identified with the transformation used in
\cite{BRGTV,HIM-M2toD2rev} to fix the value of the coupling
constant in the Bagger--Lambert theory.

As we will now show, the connected component of $\Aut^0 V$ consists of
the inner automorphisms obtained by exponentiating the inner derivations
of the Lie 3-algebra $V$.  Indeed, the inner derivations are thus
given by $\ad_{u,x} = \ad_x$, $\ad_{x,y} = [x,y]\otimes v^\flat - v
\otimes [x,y]^\flat$ for all $x,y\in\fg$, with $\flat : V \to V^*$
denoting the musical isomorphism induced by the inner product; that is,
$x^\flat(y) = \left<x,y\right>$.  Since $\ad_{x,y}$ only depends on
$x,y$ via their Lie bracket, and since $\fg = [\fg,\fg]$, we see that
the image of the $\ad_{x,y}$ is the abelian subalgebra of $\fgl(V)$
given by
\begin{equation*}
  \fg_{\text{ab}} := \left\{ t \otimes v^\flat - v \otimes t^\flat
    \middle | t\in\fg   \right\}~.
\end{equation*}
Similarly, the image of the $\ad_{u,x}$ is the adjoint Lie algebra
$\ad\fg$ of $\fg$, and it is clear that $\ad\fg$ acts on
$\fg_{\text{ab}}$ by restricting the defining representation of
$\fgl(\fg)$.  In other words, the inner derivations of $V$ span a Lie
algebra
\begin{equation*}
  \ad V \cong \fg_{\text{ab}} \rtimes \ad\fg~.
\end{equation*}
As a subalgebra of $\fso(V)$, this is contained in the stabiliser of the
null vector $v$, with $\fg_{\text{ab}}$ acting as null rotations and
$\ad\fg$ as transverse rotations.  Exponentiating $\ad V$, we obtain the
group $\Ad V = \fg_{\text{ab}} \rtimes \Ad G$, which is an affinisation
of the adjoint group.  Indeed, exponentiating $\ad \fg$ we obtain $\Ad
G$, whereas exponentiating an element of the form $t \otimes v^\flat - v
\otimes t^\flat$, we obtain
\begin{equation*}
  1 + t \otimes v^\flat - v \otimes t^\flat - \half |t|^2 v \otimes
  v^\flat~.
\end{equation*}
In summary, elements $\psi \in \Ad V$ are parametrised by $a\in\Ad G$
and $t\in \fg$ and act by
\begin{equation}\label{eq:AdVonV}
  \begin{aligned}[m]
    \psi(v) &= v\\
    \psi(u) &= u + t - \half |t|^2 v\\
    \psi(x) &= a(x) - \left<t,a(x)\right> v~,
  \end{aligned}
\end{equation}
whence $\Ad V$ is precisely the connected component of the identity of
$\Aut^0 V$, as claimed.

The Lie algebra $\Der V$ of $\Aut V$ consists of derivations of $V$.
It is isomorphic to the real Lie algebra with generators $D$, $S$, $L_x$
and $T_x$ for $x\in\fg$, subject to the following nonzero Lie brackets:
\begin{equation*}
  [D,S] = -4 S~, \qquad
  [D,T_x] = - 2 T_x~, \qquad
  [L_x,L_y] = L_{[x,y]} \qquad\text{and}\qquad
  [L_x,T_y] = T_{[x,y]}~.
\end{equation*}
If we let $\fa$ denote the two-dimensional solvable Lie subalgebra
spanned by $D$ and $S$, then we find that $\Der V$ has the following
structure
\begin{equation*}
  \Der V \cong \fa \ltimes \ad V~.
\end{equation*}

\subsection{Maximal abelian subalgebras}
\label{sec:maximal-abelian}

We now determine the maximal abelian subalgebras of $V$.  Every maximal
abelian subalgebra contains the centre $Z = \RR v$.  Let $A < V$ be a
maximal abelian subalgebra.  Then the restriction of the inner product
to $A$ is either degenerate or nondegenerate.  If degenerate, it means
that there can be no element of $A$ of the form $u + \cdots$, whereas if
it is nondegenerate, there is such an element, which can be taken to
have the form $u + z$, for some $z \in \fg$.

In the nondegenerate case, $A = \RR(u+z) \oplus \RR v \oplus B$, where
$B \subset \fg$ is a subspace obeying $[x,y] = 0$ for all $x,y\in B$.
In other words, $B$ is an abelian Lie subalgebra of $\fg$.  Since $\fg$
is compact, maximal abelian Lie subalgebras coincide with the Cartan
subalgebras.  Hence nondegenerate maximal abelian subalgebras of $V$ are
of the form $A = \RR (u+z) \oplus \RR v \oplus \fh$, for some $z \in
\fg$ and some Cartan subalgebra $\fh < \fg$.

The degenerate case is slightly more involved.  Here $A = \RR v \oplus
\fp$, where $\fp \subset \fg$ is a subspace of $\fg$ on which the
three-form $\Omega(x,y,z) = \left<[x,y],z\right>$ vanishes identically.
We call such subspaces \textbf{$\Omega$-isotropic}.  An equivalent
condition for a subspace $\fp$ to be $\Omega$-isotropic is that $[\fp,\fp]
\subset \fp^\perp$, whence the $\Omega$-isotropic Lie subalgebras are
necessarily abelian.  The maximal $\Omega$-isotropic subalgebras are
therefore the Cartan subalgebras.  However there is no need for $\fp$ to
be a Lie subalgebra: it is $A$ which has to be an abelian (3-)subalgebra
of $V$ and this only requires $\fp$ to be a subspace.  We say that an
$\Omega$-isotropic subspace is \textbf{maximal}, if it is not properly
contained in any $\Omega$-isotropic subspace.  The following is a useful
characterisation of maximality.

\begin{lemma}\label{le:maximality}
  An $\Omega$-isotropic subspace $\fp \subset \fg$ is maximal if and only if
  $[\fp, \fp] = \fp^\perp$.
\end{lemma}

\begin{proof}
  Let $\fp\subset \fg$ be an $\Omega$-isotropic subspace properly contained
  in another $\Omega$-isotropic subspace.  Then there is some $x \in
  \fp^\perp$ such that $\hat\fp = \fp \oplus \RR x$ is $\Omega$-isotropic.
  This condition is equivalent to $x \in [\fp,\fp]^\perp$ or dually that
  $[\fp,\fp] \subset x^\perp$.  In other words, $[\fp,\fp] \subsetneq
  \fp^\perp$.
\end{proof}

A large class of maximally $\Omega$-isotropic subspaces are in
one-to-one correspondence with the compact riemannian symmetric spaces.
Indeed, let $\fk \subset \fg$ be a Lie subalgebra and consider $\fp =
\fk^\perp$.  Since $\fk$ preserves the inner product, $\fp$ is stable
under the adjoint action of $\fk$; that is, $[\fk,\fp] \subset \fp$,
whence the split $\fg = \fk \oplus \fp$ is reductive.  The split will be
symmetric, so that $[\fp,\fp] \subset \fk$ precisely when $\fp$ is
$\Omega$-isotropic.  Indeed, $\Omega$ is essentially the torsion of the
canonical connection on $G/K$ and precisely when $G/K$ is a symmetric
space, this connection agrees with the Levi-Cività connection, which is
torsionless.

\begin{table}[h!]
  \centering
  \begin{tabular}{|c|>{$}c<{$}|>{$}c<{$}|>{$}c<{$}|>{$}c<{$}|>{$}c<{$}|}
    \hline
    Type & \fg & \fk & \dim\fg & \dim\fp & \text{rank}\\
    \hline\hline
    \textsf{A}I & \fsu(n) & \fso(n) & n^2-1 & \half (n-1)(n+2) & n-1\\
    \textsf{A}II & \fsu(2n) & \fsp(n) & 4n^2 - 1 & (n-1)(2n+1) & n-1\\
    \textsf{A}III & \fsu(p+q) & \fsu(p)\oplus\fsu(q) \oplus \RR & (p+q)^2-1 & 2pq & \min(p,q)\\
    \textsf{BD}I & \fso(p+q) & \fso(p) \oplus \fso(q) & \half (p+q)(p+q-1) & pq  & \min(p, q)\\
    \textsf{D}III & \fso(2n) & \fu(n) & n(2n-1) & n(n-1) & \lfloor n/2 \rfloor \\
    \textsf{C}I & \fsp(n) & \fu(n) & n(2n+1) & n(n+1)  & n \\
    \textsf{C}II & \fsp(p+q) & \fsp(p) \oplus \fsp(q) & (p+q)(2p+2q+1) & 4pq & \min(p,q)\\
    \textsf{E}I & \fe_6 & \fsp(4) & 78 & 42 & 6\\
    \textsf{E}II & \fe_6 & \fsu(6) \oplus \fsu(2) & 78 & 40 & 4\\
    \textsf{E}III & \fe_6 & \fso(10) \oplus \fso(2) & 78 & 32 & 2\\
    \textsf{E}IV & \fe_6 & \ff_4 & 78 & 26 & 2\\
    \textsf{E}V & \fe_7 & \fsu(8) & 133 & 70 & 7\\
    \textsf{E}VI & \fe_7 & \fso(12)\oplus \fsu(2) & 133 & 64 & 4\\
    \textsf{E}VII & \fe_7 & \fe_6 \oplus \fso(2) & 133 & 54 & 3\\
    \textsf{E}VIII & \fe_8 & \fso(16) & 248 & 128 & 8\\
    \textsf{E}IX & \fe_8 & \fe_7\oplus \fsu(2) & 248 & 112 & 4\\
    \textsf{F}I & \ff_4 & \fsp(3) \oplus \fsu(2) & 52 & 28 & 4\\
    \textsf{F}II & \ff_4 & \fso(9) & 52 & 16 & 1\\
    \textsf{G} & \fg_2 & \fso(4) & 14 & 8 & 2\\\hline
  \end{tabular}
  \vspace{8pt}
  \caption{Symmetric splits $\fg = \fk \oplus \fp$ of compact simple Lie algebras}
  \label{tab:RSS}
\end{table}

There are two types of compact irreducible riemannian symmetric spaces:
\begin{itemize}
\item \textbf{Type I}: $G/K$ where $G$ is a compact simple Lie group and
  $K$ a subgroup with Lie algebra $\fk < \fg$ such that orthogonal
  decomposition $\fg = \fk \oplus \fp$ is a symmetric split; and
\item \textbf{Type II}: compact simple Lie groups $H$ relative to a
  bi-invariant metric.  This can be written in terms of a symmetric
  split, with $\fg = \fh \oplus \fh$, with the ad-invariant inner
  product on both copies of $\fh$ being the same, and $\fk =
  \left\{(x,x) \middle | x \in \fh \right\}$ the diagonal Lie
  subalgebra.  Its perpendicular complement is $\fp = \left\{(x,-x)
    \middle | x \in \fh \right\}$.  It is easy to see that
  $[\fp,\fp]=\fk$, whence by Lemma~\ref{le:maximality}, $\fp$ is
  maximal.
\end{itemize}

Table~\ref{tab:RSS} lists the type I irreducible riemannian symmetric
spaces in terms of their symmetric splits $\fg = \fk\oplus \fp$ using
the Cartan nomenclature as described in \cite{Helgason}.  There are some
repetitions in the table, which can be eliminated by taking $p\geq q$ in
\textsf{A}III, \textsf{BD}I and \textsf{C}II; $n\geq 2$ in \textsf{A}I,
\textsf{A}II and \textsf{C}I; $n\geq 5$ in \textsf{D}III; $p\geq 2$ in
\textsf{A}III; and taking $p+q\geq 7$ in \textsf{BD}I in addition to
$p=q=1$.

\begin{proposition}
  The $\Omega$-isotropic subspaces $\fp=\fk^\perp$ in Table~\ref{tab:RSS} are
  maximal.
\end{proposition}

\begin{proof}
  We observe that the Jacobi identity says that $[\fp,\fp]$ is an ideal
  of $\fk$.  Now many of the $\fk$ in Table~\ref{tab:RSS} are simple,
  whence $[\fp,\fp]$, being nonzero, must be all of $\fk$ and by
  Lemma~\ref{le:maximality}, $\fp$ is maximal.  For the remaining
  entries but one, $\fk = \fk_1 \oplus \fk_2$, with $\fk_i$ a simple or
  one-dimensional ideal.  Then $\fp$, if not maximal, must satisfy
  $[\fp,\fp] = \fk_i$ for some $i=1,2$.  Let's assume, without loss of
  generality, that $[\fp,\fp]=\fk_1$.  Then
  $\left<[\fp,\fp],\fk_2\right>=0$, but this means that
  $\left<[\fk_2,\fp],\fp\right>=0$, whence $[\fk_2,\fp]=0$, since it
  belongs to $\fp$.  Since $[\fk_1,\fk_2]=0$, $\fk_2$ would be an ideal
  of $\fg$, contradicting the fact that $\fg$ is simple.  Therefore
  $[\fp,\fp]=\fk$ and again by Lemma \ref{le:maximality} it is maximal.
  Finally, in case \textsf{A}III, $\fk = \fsu(p) \oplus \fsu(q) \oplus
  \RR$.  In the same as when $\fk$ is the sum of two ideals, one shows
  that $[\fp,\fp]^\perp \cap \fk$ is an ideal of $\fg$.  Since it cannot
  be all of $\fg$, simplicity of $\fg$ says that it must be zero.
\end{proof}

We have not been able to construct any maximal $\Omega$-isotropic
subspace $\fp \subset \fg$ which does not come from a symmetric split,
but neither have we been able to prove that they all arise in this way;
although it would be tempting to conjecture that this is the case.

\subsection{The moduli spaces}
\label{sec:moduli-spaces}

We now quotient by inner automorphisms to arrive at the moduli spaces.

The classical space of vacua has two main branches, corresponding to
degenerate and nondegenerate maximal abelian subalgebras $A<V$.  The
degenerate branch splits further into sub-branches labelled by the
different types of maximal $\Omega$-isotropic subspaces, many of which
are given by (not necessarily irreducible) compact symmetric spaces.

\subsubsection{Moduli of nondegenerate maximal abelian subalgebras}
\label{sec:nondegmod}

Let us first consider the branch of the moduli space corresponding to
nondegenerate maximal abelian subalgebras of $V$ of the form $A(\fh,z)
= \RR(u + z) \oplus \RR v \oplus \fh$ for some $z\in\fg$ and some
Cartan subalgebra $\fh < \fg$.  Such maximal abelian subalgebras are
parametrised by $\Ca(\fg) \times \fg$, where $\Ca(\fg)$ is the space of
Cartan subalgebras of $\fg$, which we can think of as a submanifold of
the grassmannian of $\rank \fg$-planes in $\fg$, isometric to $G/H$,
where $H$ is the maximal torus of a fixed Cartan subalgebra.

The subgroup $\Ad V < \GL(V)$ acts on $V$ and hence on the $A(\fh,z)$.
Let $\psi = \psi(t,a) \in \Ad V$ with $t \in \fg$ and $a \in \Ad G$.
Then from \eqref{eq:AdVonV} we see that, for all $x\in\fg$,
\begin{align*}
  \psi(v) &= v\\
  \psi(u+z) &= u + a z + t - \left(\left<t,a z\right> + \half
    |t|^2\right) v\\
  \psi(x) &= a x - \left<t,a x\right>v~.
\end{align*}
In other words, $\psi$ maps the subspace $A(\fh,z)$ to $A(a\fh, a z +
t)$.  Now fix a Cartan subalgebra $\fh_0 \subset \fg$ and let $A_0 =
A(\fh_0,0)$.  Any other Cartan subalgebra of $\fg$ can be obtained from
$\fh_0$ by acting with some $a \in \Ad G$.  The translational component
of $\Ad V$ allows us to shift $z$ in $(\fh,z)$ to $0$.  In other words,
given any $A(\fh,z)$, there is some $\psi\in\Ad V$ such that $A(\fh,z) =
\psi A_0$.

The subset of classical vacua corresponding to such maximal subalgebras
is given by
\begin{align*}
  \eV_{\text{nondeg}} &= \bigcup_{(\fh,z)\in\Ca(\fg)\times\fg} \Hom(\RR^8,A(\fh,z))\\
  &= \bigcup_{\psi\in\Ad V} \Hom(\RR^8,\psi A_0)\\
  &= \bigcup_{\psi\in\Ad V} \psi \Hom(\RR^8,A_0)~.
\end{align*}
In other words, it is the orbit of $\Hom(\RR^8,A_0)$ under $\Ad V$.
Quotienting by $\Ad V$ yields
\begin{equation*}
  \eM_{\text{nondeg}} = \Hom(\RR^8,A_0)/G_0
\end{equation*}
where $G_0$ is the stabiliser of $A_0$ (and hence of $\Hom(\RR^8,A_0)$)
in $\Ad V$.  This quotient is not trivial because there are elements of
$\Ad V$ which preserve $A_0$ as a subspace, but not $A_0$ pointwise.
Indeed, the stabiliser of $A_0$ in $\Ad V$ is the same as the stabilizer
of $(\fh_0,0) \in \Ca(\fg)\times \fg$.  The translations move the
origin, hence $G_0$ is a subgroup of $\Ad G$.  In fact, it is $\Ad
N(H_0)$, where $N(H_0)$ is the normaliser (in $G$) of the maximal torus
$H_0$ corresponding to $\fh_0$.  Clearly $H_0 < N(H_0)$ fixes every
point of $\fh_0$, whence only the Weyl group $\fW_0 := N(H_0)/H_0$ of
$\fh_0$ acts effectively.

In summary, the branch of the moduli space of classical vacua
corresponding to nondegenerate maximal abelian subalgebras of $V$ is
given by
\begin{equation*}
  \eM_{\text{nondeg}} = \Hom(\RR^8,\RR v \oplus \RR u \oplus \fh_0)/\fW_0~,
\end{equation*}
where $\fW_0$ is the Weyl group of $\fh_0$.  We see that
$\dim\eM_{\text{nondeg}} = 8(2 + \rank\fg)$.  In the next section we
will study the asymptotics of this branch for large rank as $\fg$ varies
among the compact semisimple Lie algebras.

\subsubsection{Moduli of degenerate maximal abelian subalgebras}
\label{sec:degmod}

We now consider the branch of the moduli space corresponding to
degenerate maximal abelian subalgebras of $V$ of the form $A(\fp) = \RR
v \oplus \fp$, where $\fp \subset \fg$ is a subspace obeying
$[\fp,\fp]=\fp^\perp$.  In other words, such subalgebras are
parametrised by the space $\Iso_\Omega(\fg)$ of maximal
$\Omega$-isotropic subspaces of $\fg$.  The adjoint group $\Ad G$
preserves $\Iso_\Omega(\fg)$ and decomposes it into orbits.  Some of
these orbits are in one-to-one correspondence with (isometry classes of)
compact riemannian symmetric spaces.  The group $\Ad V$ acts on the
$A(\fp)$ as follows.  Let $\psi=\psi(a,t) \in \Ad V$ with $a \in \Ad G$
and $t \in \fg$.  Then from \eqref{eq:AdVonV} we see that $\psi(v) = v$
and, for $x\in\fg$, and $\psi(x) = a x - \left<t,ax\right>v$, whence
$\psi A(\fp) = A(a\fp)$.

The subset of classical vacua corresponding to the $A(\fp)$ is given by
\begin{align*}
  \eV_{\text{deg}} &= \bigcup_{\fp \in \Iso_\Omega(\fg)}
  \Hom(\RR^8,A(\fp))\\
  &= \bigcup_{[\fp_0]\in\Iso_\Omega(\fg)/\Ad G} \quad \bigcup_{\psi \in
    \Ad V} \Hom(\RR^8,\psi A(\fp_0))~,
\end{align*}
where $\fp_0$ stands for a representative subspace in the orbit
$[\fp_0]$ of $\Ad G$ on $\Iso_\Omega(\fg)$.  A subset of orbits consists
of isometry classes of compact riemannian symmetric spaces of the form
$G/K$ with $\fp_0= \fk^\perp$.

Each orbit $[\fp_0]$ gives rise to a branch of the moduli space obtained
by quotienting the $\Ad V$ orbit of $\Hom(\RR^8,A(\fp_0)$ by $\Ad V$.
As before, the result is
\begin{equation*}
  \eM_{[\fp_0]}  = \Hom(\RR^8,A(\fp_0)/(\fg_{\text{ab}} \rtimes \Ad K_0)~,
\end{equation*}
where $\Ad K_0 < \Ad G$ is the stabiliser of $\fp_0$.

Every (isometry class of) compact riemannian symmetric space gives rise
to one such $\eM_{[\fp_0]}$ with $\dim\fp_0$ being the dimension of the
symmetric space.  Those symmetric spaces which are products of
irreducibles of type \textsf{A}I, \textsf{BD}I (with $p=q$),
\textsf{C}I, \textsf{E}I, \textsf{E}V, \textsf{E}VIII, \textsf{F}I and
\textsf{G} have maximal dimension for a given $\fg$: their dimension
being $\half(\dim\fg + \rank\fg)$.  They are characterised by the
property that $\fp$ contains a Cartan subalgebra of $\fg$, or
equivalently, that they are of maximal rank.  At the other extreme,
symmetric spaces which are products of spheres, complex and quaternionic
projective spaces (i.e., types \textsf{A}III, \textsf{BD}I and
\textsf{C}II, all with $q=1$) have smallest possible dimension compared
with the dimension of $\fg$.  Whereas in the former class the dimension
of the moduli space of vacua grows like the square of the rank of $\fg$,
in the latter the dimension of the moduli space grows linearly.

It would be very interesting to find a natural interpretation of these
maximal $\Omega$-isotropic subspaces in M-theory.  Whereas the
non-degenerate branch seems to capture the expected vacua of the super
Yang-Mills description on a stack of D2-branes (perhaps with an
orientifold plane) lifted to M-theory, the configuration space of
M-branes that might mimic the degenerate branch is less clear to us.

\subsection{Asymptotic behaviour}
\label{sec:asymptotia}

It is expected \cite{KlebanovTseytlinEntropy,KlebanovAbsorption} that
for a theory of $N$ coincident M2-branes, the number of physical
degrees of freedom should grow as $N^{3/2}$ for large $N$.  In the
present context, the number $N$ is the dimension of the moduli space
of vacua, whereas the number of degrees of freedom is the dimension of
the Lie 3-algebra $V$.  For the nondegenerate branch we have seen that
whereas $\dim V = \dim \fg + 2$, the dimension of the moduli space of
vacua is $8(\rank \fg + 2)$.  It is therefore natural to ask how
$(\dim\fg,\rank\fg)$ are distributed for semisimple Lie algebras.

\begin{table}[h!]
  \centering
  \begin{tabular}{|>{$}c<{$}|>{$}c<{$}|>{$}c<{$}|>{$}c<{$}|>{$}c<{$}|}
    \hline
    \text{Lie algebra} & \text{Rank} & \text{Dimension} & \text{Weyl
      group $\fW$} & \text{order of $\fW$}\\
    \hline\hline
    A_n & n & n(n+2) & S_{n+1}  & (n+1)!\\
    B_n & n & n(2n+1) & (\ZZ_2)^n \rtimes S_n & 2^n n!\\
    C_n & n & n(2n+1) &  (\ZZ_2)^n \rtimes S_n & 2^n n!\\
    D_n & n & n(2n-1) &  (\ZZ_2)^{n-1} \rtimes S_n & 2^{n-1} n!\\
    E_6 & 6 & 78 &  & 2^7 3^4 5\\
    E_7 & 7 & 133 & & 2^{10} 3^4 5~7\\
    E_8 & 8 & 248 &  & 2^{14} 3^5 5^2 7\\
    F_4 & 4 & 52 & S_3 \ltimes (S_4 \ltimes (\ZZ_2)^3)& 2^7 3^2\\
    G_2 & 2 & 14 & D_6 & 2^2 3\\
    \hline
  \end{tabular}
  \vspace{8pt}
  \caption{The simple Lie algebras}
  \label{tab:simple}
\end{table}

Table~\ref{tab:simple} lists the ranks and dimensions of the simple Lie
algebras, as well as information on the Weyl group which may help to
interpret the classical vacua in terms of configurations of M2 branes.
This data allows us to write the following partition function for
compact semisimple Lie algebras:
\begin{multline*}
  Z_{\text{SSLA}}(t,\omega) = \prod_{n\geq 1} \frac{1}{1-\omega^n t^{n(n+2)}}
  \prod_{n\geq 2} \frac{1}{1-\omega^n t^{n(2n+1)}}
  \prod_{n\geq 3} \frac{1}{1-\omega^n t^{n(2n+1)}}
  \prod_{n\geq 4} \frac{1}{1-\omega^n t^{n(2n-1)}}\\
  \times 
  \frac{1}{1-\omega^6 t^{78}}
  ~ 
  \frac{1}{1-\omega^7 t^{133}}
  ~ 
  \frac{1}{1-\omega^8 t^{248}}
  ~ 
  \frac{1}{1-\omega^4 t^{52}}
  ~ 
  \frac{1}{1-\omega^2 t^{14}}~.
\end{multline*}
Expanding as a power series in $t$ and $\omega$, we have
\begin{equation*}
  Z_{\text{SSLA}}(t,\omega) = 1 + \sum_{d,r\geq 1} N_{d,r} t^d \omega^r~,
\end{equation*}
where $N_{d,r}$ is the number of $d$-dimensional compact semisimple
Lie algebras of rank $r$.  A little computer experimentation shows
that the ranks are normally distributed around the mean rank, which
grows linearly with dimension.  Figure \ref{fig:averagerank} shows a
plot of the average rank as a function of dimension for the
$23,058,218,050,191,608$ compact semisimple Lie algebras of dimension
$\leq 1000$, depicted by the blue line.  The red line is the graph of $r
= d^{2/3}$.

\begin{figure}[h!]
  \centering
  \includegraphics[width=0.75\textwidth]{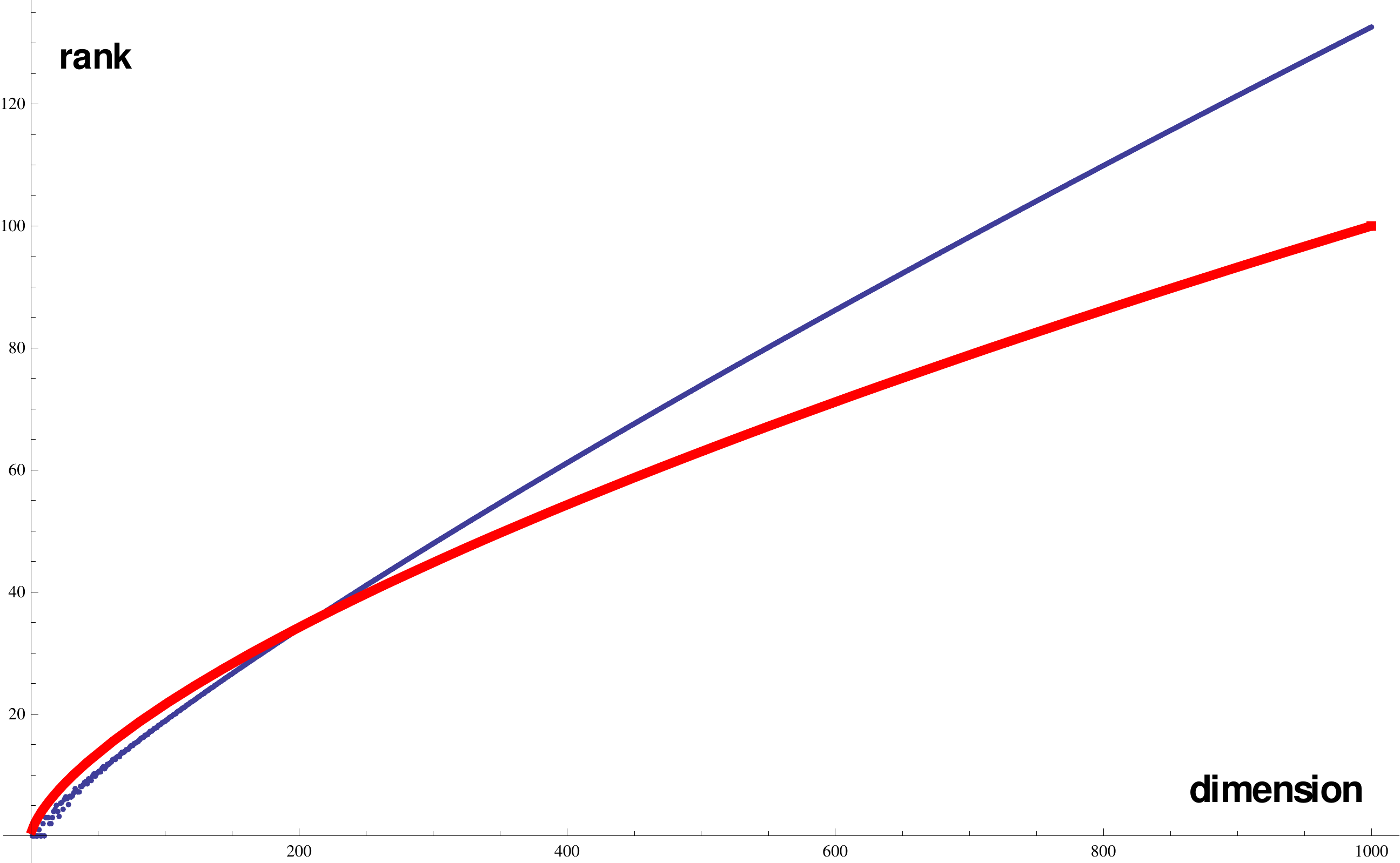}
  \vspace{8pt}
  \caption{Average rank as a function of dimension for compact semisimple Lie algebras}
  \label{fig:averagerank}
\end{figure}

Despite not being the generic behaviour, it is not difficult to come up
with series of semisimple Lie algebras whose rank and dimension obey $d
= r^{3/2}$.  Indeed, the classical simple Lie algebras of rank $n$ have
dimension which goes like $n^2$ for large $n$.  Hence taking $n$ such
algebras yields a sequence of semisimple Lie algebras where $d \sim n^3$
and $r \sim n^2$ for large $n$.  Indeed, let $\fg_n = \fs_1 \oplus \dots
\oplus \fs_n$ where the $\fs_i$ are any one of $A_n$, $B_n$, $C_n$ or
$D_n$ and let $V_n$ denote the indecomposable lorentzian Lie 3-algebra
constructed out of $\fg_n$.  The nondegenerate branch of the moduli
space of the Bagger--Lambert models associated to $V_n$ exhibit the
desired behaviour between the number of M2-branes ($\frac18 \dim\eM$)
and the number of degrees of freedom ($\dim V$).

For the degenerate branch of the moduli space, the relevant question is
how $(\dim\fg,\dim\fp)$ are distributed.  Using Table~\ref{tab:RSS} for
the Type I irreducible riemannian symmetric spaces and again
Table~\ref{tab:simple} for the ones with Type II, and being careful to
avoid repetitions due to low-dimensional isomorphisms, it is possible
to write down the following generating functional for compact riemannian
symmetric spaces
\begin{equation*}
  Z_{\text{CRSS}}(t,\omega) = Z_{\text{type I}}(t,\omega) Z_{\text{type
      II}}(t,\omega)~,
\end{equation*}
where
\begin{multline*}
  Z_{\text{type I}}(t,\omega) = \prod_{n\geq 2} \frac{1}{1 - \omega^{(n-1)(n+2)} t^{n^2-1}}
  \prod_{n\geq 2} \frac{1}{1 - \omega^{(n-1)(2n+1)} t^{4n^2-1}}
  \prod_{\substack{p\geq q \geq 1\\ p\geq 2}} \frac{1}{1 - \omega^{2pq} t^{(p+q)^2-1}}\\
  \times   \prod_{n\geq 2} \frac{1}{1 - \omega^{(n-1)(n+2)} t^{n^2-1}}
  \prod_{n\geq 2} \frac{1}{1 - \omega^{n(n-1)} t^{n(2n+1)}}
  \prod_{n\geq 5} \frac{1}{1 - \omega^{n(n-1)} t^{n(2n-1)}}\\
  \times \prod_{\substack{p\geq q \geq 1\\ p+q \geq 7}} \frac{1}{1 - \omega^{pq} t^{(p+q)(p+q-1)/2}}
  \prod_{p\geq q \geq 1} \frac{1}{1 - \omega^{4pq} t^{(p+q)(2p+2q+1)}}\\
  \times
  \frac{1}{1 - \omega t}
  ~
  \frac{1}{1 - \omega^{42} t^{78}} 
  ~
  \frac{1}{1 - \omega^{40} t^{78}} 
  ~
  \frac{1}{1 - \omega^{32} t^{78}} 
  ~
  \frac{1}{1 - \omega^{26} t^{78}} 
  ~
  \frac{1}{1 - \omega^{70} t^{133}} 
  ~
  \frac{1}{1 - \omega^{64} t^{133}} \\
  \times
  \frac{1}{1 - \omega^{54} t^{133}} 
  ~
  \frac{1}{1 - \omega^{128} t^{248}} 
  ~
  \frac{1}{1 - \omega^{112} t^{248}} 
  ~
  \frac{1}{1 - \omega^{28} t^{52}} 
  ~
  \frac{1}{1 - \omega^{16} t^{52}} 
  ~
  \frac{1}{1 - \omega^8 t^{14}}
\end{multline*}
is the partition function corresponding to Type I riemannian symmetric
spaces and
\begin{multline*}
   Z_{\text{type II}}(t,\omega) = \prod_{n\geq 1} \frac{1}{1- (\omega t^2)^{n(n+2)}}
   \prod_{n\geq 2} \frac{1}{1- (\omega t^2)^{n(2n+1)}}
   \prod_{n\geq 3} \frac{1}{1- (\omega t^2)^{n(2n+1)}}\\
   \times \prod_{n\geq 4} \frac{1}{1- (\omega t^2)^{n(2n-1)}}
   ~
   \frac{1}{1- (\omega t^2)^{78}}
   ~ 
   \frac{1}{1- (\omega t^2)^{133}}
   ~ 
   \frac{1}{1- (\omega t^2)^{248}}
   ~ 
   \frac{1}{1- (\omega t^2)^{52}}
   ~ 
   \frac{1}{1- (\omega t^2)^{14}}
\end{multline*}
is the corresponding to Type II riemannian symmetric  spaces.
Expanding as a power series in $t$ and $\omega$, we have
\begin{equation*}
  Z_{\text{CRSS}}(t,\omega) = 1 + \sum_{d,s\geq 1} N_{d,s} t^s \omega^d~,
\end{equation*}
where $N_{d,s}$ is now the number of $d$-dimensional compact riemannian
symmetric spaces with an $s$-dimensional group of isometries.
Figure \ref{fig:averagedimension} shows a
plot of the average dimension of a compact riemannian symmetric space
$(\dim\fp)$ as a function of the dimension  of its isometry group
$(\dim\fg$) for the $378,683,913,003,348,073,310,000,493,022$ compact
riemannian symmetric spaces whose isometry group have dimension $\leq
1000$, depicted by the blue line.  The red line is the graph of $\dim\fp
= (\dim\fg)^{2/3}$.

\begin{figure}[h!]
  \centering
  \includegraphics[width=0.75\textwidth]{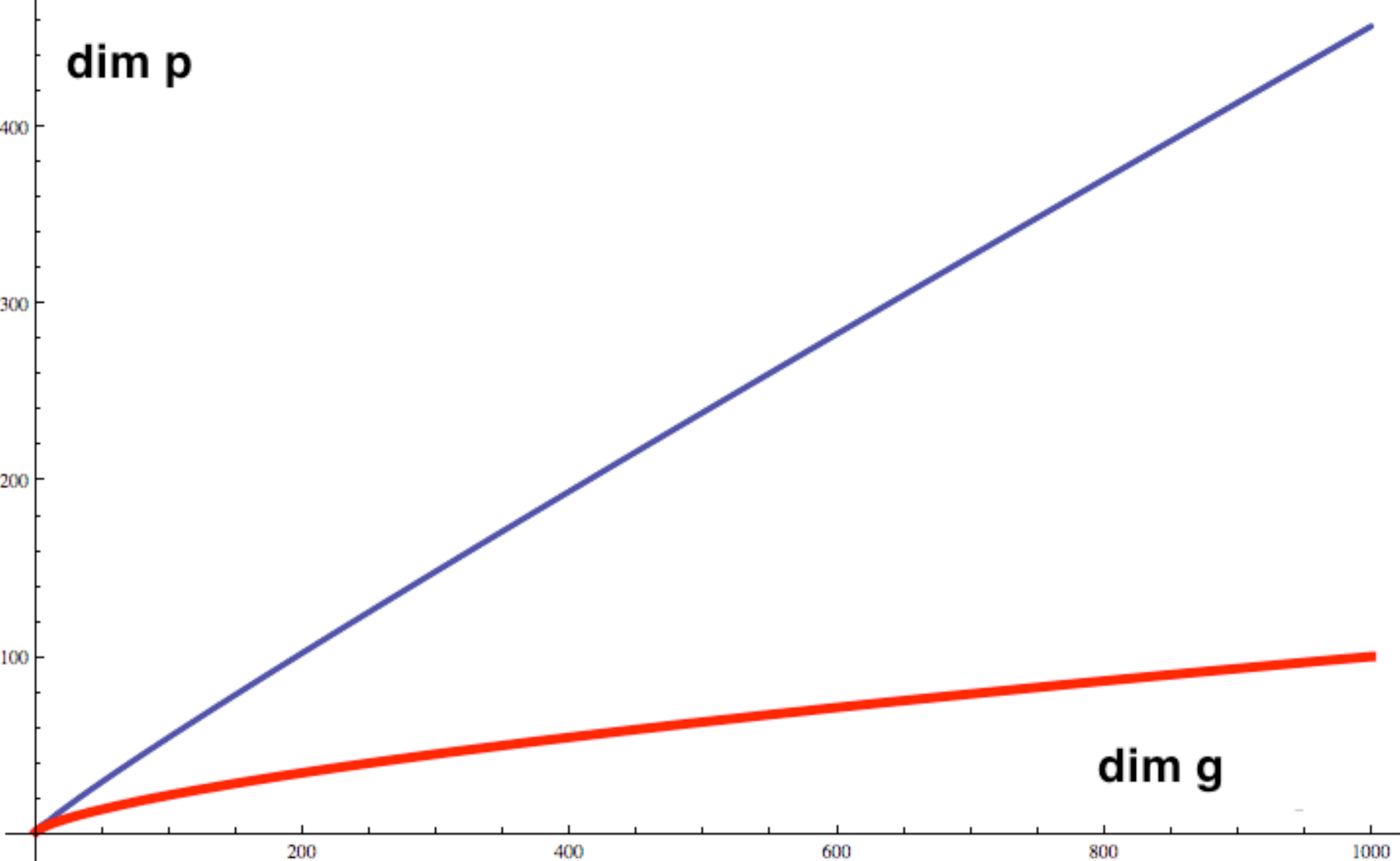}
  \vspace{8pt}
  \caption{Average $\dim\fp$ as a function of $\dim\fg$, for
    $(\fg,\fp^\perp)$ a symmetric split}
  \label{fig:averagedimension}
\end{figure}

Despite the fact that the average behaviour is linear, just as in the
case of the nondegenerate branch of the moduli space, it is easy to cook
up models where the behaviour is of the desired form.  Consider the type
I symmetric spaces \textsf{BD}I, \textsf{A}III and \textsf{C}II of rank
1; that is, those with $q=1$.  The dimension of the isometry group
clearly goes like the square of the dimension of the symmetric space:
$\dim\fg = n(n+1)/2$ and $\dim\fp = n$ for \textsf{BD}I; $\dim\fg =
n(n+1)$ and $\dim\fp = 2n$ for \textsf{A}III; and $\dim\fg =
(n+1)(2n+3)$ and $\dim\fp = 4n$ for \textsf{C}II.  Taking a product of
$n$ such symmetric spaces, yields $\dim\fp \sim n^2$ whereas $\dim\fg
\sim n^3$.  In fact, one is not restricted to rank 1 symmetric spaces.
Taking $p$ to infinity while keeping $q$ fixed in the above cases also
yields the right asymptotic behaviour.

A concrete series of models which realises this behaviour is to take the
($n(n^2-1)+2$)-dimensional lorentzian Lie 3-algebra $V_n$ in Theorem
\ref{th:lorentzian} corresponding to the semisimple Lie algebra
\begin{equation*}
  \fg_n = \underbrace{\fsu(n) \oplus \dots \oplus \fsu(n)}_{\text{$n$
      times}}~,
\end{equation*}
and a choice of scale for the inner product on each simple factor.  The
classical moduli space of the Bagger--Lambert model associated to $V_n$
has (at least) the following branches:
\begin{itemize}
\item a \emph{nondegenerate} branch, where the moduli space becomes
  \begin{equation*}
    \RR^{16} \times \RR^{8 n(n-1)}/(S_n)^n~,
  \end{equation*}
  which has dimension $8(2 + n(n-1))$;
\item a number of \emph{degenerate} branches, associated to the compact
  riemannian symmetric spaces with isometry $\fg_n$; that is, products
  of the following irreducible factors:
  \begin{itemize}
  \item the type II symmetric space $(\fsu(n) \oplus \fsu(n),
    \fsu(n))$, of dimension $n^2-1$;
  \item the type I symmetric space \textsf{A}I, of dimension
    $\half (n-1)(n+2)$;
  \item if $n$ is even, the type I symmetric space
    \textsf{A}II, of dimension $n^2-1$; and
  \item the type I symmetric spaces \textsc{A}III, of dimension
    $2q(n-q)$ for $1\leq q \leq n-1$.
  \end{itemize}
\end{itemize}
In particular, both the nondegenerate branch and the degenerate branch
consisting of $n$ type I symmetric spaces \textsf{A}III with $q=1$ have
the right asymptotic behaviour.

In summary, we believe that the question now is not whether there exist
plausible Bagger--Lambert models with the $N^{3/2}$ asymptotic
behaviour, but how to choose among the plethora of such models.

\bibliographystyle{utphys}
\bibliography{AdS,AdS3,ESYM,Sugra,Geometry,Algebra}

\end{document}